\documentclass[11pt]{article}%
\usepackage{amsmath}
\usepackage{amsfonts}
\usepackage{amssymb}
\usepackage{graphicx}
\usepackage{fullpage}
\usepackage{tikz}%
\providecommand{\U}[1]{\protect\rule{.1in}{.1in}}
\usetikzlibrary{plotmarks}
\definecolor{DarkBlue}{rgb}{0,0,1}
\definecolor{DarkPink}{rgb}{1,0,0}
\definecolor{Black}{rgb}{0,0,0}
\newtheorem{theorem}{Theorem}

\newtheorem{corollary}[theorem]{Corollary}

\newtheorem{lemma}[theorem]{Lemma}

\newtheorem{proposition}[theorem]{Proposition}

\newenvironment{proof}[1][Proof]{\noindent\textbf{#1.} }{\ \rule{0.5em}{0.5em}}
\begin{document}

\title{BosonSampling Is Far From Uniform}
\author{Scott Aaronson\thanks{MIT. \ Email: aaronson@csail.mit.edu. \ \ This material
is based upon work supported by the National Science Foundation under Grant
No. 0844626,\ an Alan T. Waterman Award, a Sloan Fellowship, a TIBCO Chair,
and an NSF STC on Science of Information.}
\and Alex Arkhipov\thanks{MIT. \ Email: arkhipov@mit.edu. \ Supported by an NSF
Graduate Fellowship.}}
\date{}
\maketitle

\begin{abstract}
BosonSampling, which we proposed three years ago, is a scheme for using
linear-optical networks to solve sampling problems that appear to be
intractable for a classical computer. \ In a recent manuscript, Gogolin et
al.\ claimed that even an ideal BosonSampling device's output would be
\textquotedblleft operationally indistinguishable\textquotedblright\ from a
uniform random outcome, at least \textquotedblleft without detailed a priori
knowledge\textquotedblright; or at any rate, that telling the two apart might
itself be a hard problem. \ We first answer these claims---explaining why the
first is based on a definition of \textquotedblleft a priori
knowledge\textquotedblright\ so strange that, were it adopted, almost no
quantum algorithm could be distinguished from a pure random-number source;
while the second is neither new nor a practical obstacle to interesting
BosonSampling experiments. \ However, we then go further, and address some
interesting research questions inspired by Gogolin et al.'s mistaken
arguments. \ We prove that, with high probability over a Haar-random matrix
$A$, the BosonSampling distribution induced by $A$ is far from the uniform
distribution in total variation distance. \ More surprisingly, and directly
counter to Gogolin et al., we give an efficient algorithm that distinguishes
these two distributions with constant bias. \ Finally, we offer three
\textquotedblleft bonus\textquotedblright\ results about BosonSampling.
\ First, we report an observation of Fernando Brandao: that one can
efficiently sample a distribution that has large entropy and that's
indistinguishable from a BosonSampling distribution by any circuit of fixed
polynomial size. \ Second, we show that BosonSampling distributions can be
efficiently distinguished from uniform even with photon losses and for general
initial states. \ Third, we offer the simplest known proof that
\textit{Fermion}Sampling is solvable in classical polynomial time, and we
reuse techniques from our BosonSampling analysis to characterize random
FermionSampling distributions.

\end{abstract}

\section{Background\label{INTRO}}

\textsc{BosonSampling} \cite{aark}\ can be defined as the following
computational problem. \ We are given as input an $m\times n$\ complex matrix
$A\in\mathbb{C}^{m\times n}$ ($m\geq n$), whose $n$\ columns are orthonormal
vectors in $\mathbb{C}^{m}$. \ Let $\Phi_{m,n}$\ be the set of all\ lists
$S=\left(  s_{1},\ldots,s_{m}\right)  $\ of $m$ nonnegative integers summing
to $n$ (we call these lists \textquotedblleft experimental
outcomes\textquotedblright); note that $\left\vert \Phi_{m,n}\right\vert
=\binom{m+n-1}{n}$. \ For each outcome $S\in\Phi_{m,n}$, let $A_{S}%
\in\mathbb{C}^{n\times n}$\ be the $n\times n$\ matrix that consists of
$s_{1}$\ copies of $A$'s first\ row, $s_{2}$\ copies of $A$'s second\ row, and
so on. \ Then let $\mathcal{D}_{A}$\ be the following probability distribution
over $\Phi_{m,n}$:%
\begin{equation}
\Pr_{\mathcal{D}_{A}}\left[  S\right]  =\frac{\left\vert \operatorname*{Per}%
\left(  A_{S}\right)  \right\vert ^{2}}{s_{1}!\cdots s_{m}!}, \label{per}%
\end{equation}
where $\operatorname*{Per}$\ represents the matrix permanent.\footnote{This is
indeed a normalized probability distribution; see for example \cite{aark}\ for
a proof.} \ The \textsc{BosonSampling}\ problem is to sample from
$\mathcal{D}_{A}$, either exactly or approximately. \ Ideally, we want to do
so in time polynomial in $n$ and $m$.

The \textsc{BosonSampling}\ problem has no known applications to cryptography
or anything else. \ Nevertheless, it has two remarkable properties that
motivate its study:

\begin{enumerate}
\item[(1)] \textsc{BosonSampling} is easy to solve using a quantum computer.
\ Indeed, it is solvable by an especially simple \textit{kind} of quantum
computer: one that consists entirely of a network of beamsplitters, through
which identical single photons are sent and then nonadaptively
measured.\footnote{This type of quantum computer is not believed to be
universal for quantum computation, or for that matter, even for
\textit{classical} computation! \ On the other hand, if we throw in one
additional resource---namely \textit{adaptivity}, or the ability to condition
later beamsplitter settings on the outcomes of earlier photon-number
measurements---then Knill, Laflamme, and Milburn \cite{klm}\ famously showed
that we \textit{do} get the capacity for universal quantum computation.} \ The
reason for this is trivial: while our definition of \textsc{BosonSampling}%
\ was purely mathematical, the problem directly models the physics of
identical, non-interacting photons (or in general, any bosonic particles).
\ We merely need to interpret $n$\ as the number of photons (one for each
input location or \textquotedblleft mode\textquotedblright), $m$ as the number
of output modes, $A$ as the transition matrix between the input and output
modes (as determined by the beamsplitter network), and $S=\left(  s_{1}%
,\ldots,s_{m}\right)  $\ as a possible output configuration consisting of
$s_{i}$\ photons in the $i^{th}$\ mode for each $i$. \ Then according to
quantum mechanics, the final amplitude for the basis state $\left\vert
S\right\rangle $ is $\operatorname*{Per}\left(  A_{S}\right)  /\sqrt
{s_{1}!\cdots s_{m}!}$, so the probability of observing $\left\vert
S\right\rangle $\ on measuring is given by the formula (\ref{per}).

\item[(2)] By contrast, it is \textit{not} known how to solve
\textsc{BosonSampling}\ efficiently using a classical computer. \ Indeed, one
can say something much stronger: we showed in \cite{aark} that, if there
\textit{is} a polynomial-time classical algorithm for exact
\textsc{BosonSampling}, then $\mathsf{P}^{\mathsf{\#P}}=\mathsf{BPP}%
^{\mathsf{NP}}$\ (and hence the polynomial hierarchy collapses), which is
considered vanishingly unlikely. \ The proof made indirect use of the
$\mathsf{\#P}$-completeness of the\ permanent. \ Even for \textit{approximate}
\textsc{BosonSampling}, we proved that a fast classical algorithm would imply
a $\mathsf{BPP}^{\mathsf{NP}}$\ algorithm to approximate the permanent of an
$n\times n$ matrix $X$\ of independent $\mathcal{N}\left(  0,1\right)
_{\mathbb{C}}$\ Gaussian entries,\ with high probability over $X$. \ If (as we
conjecture) this approximation problem is already $\mathsf{\#P}$-complete,
then a $\mathsf{BPP}^{\mathsf{NP}}$\ algorithm for it would be essentially
ruled out as well. \ In summary, \textsc{BosonSampling}\ lets us base our
belief in \textquotedblleft Feynman's Conjecture\textquotedblright---the
conjecture that quantum mechanics is exponentially hard to simulate by
classical computers---on assumptions that seem much more \textquotedblleft
generic\textquotedblright\ than (say) the classical hardness of factoring integers.
\end{enumerate}

One can study \textsc{BosonSampling}, as we did at first, purely from a
theoretical computer science standpoint. \ However, \textsc{BosonSampling} can
also be seen as an implicit proposal for a\ physics experiment---and perhaps
not surprisingly,\ that is what has led to most of the interest in it.

In an ideal \textsc{BosonSampling}\ experiment, one would simultaneously
generate $n$ identical single photons, one in each of $n$ input modes. \ One
would then send the photons through a large network of beamsplitters, with the
beamsplitter angles \textquotedblleft random\textquotedblright\ and
\textquotedblleft arbitrary\textquotedblright\ but known to the experimenter
in advance. \ Finally, one would measure the number of photons in each of $m$
output modes, and check (after sufficiently many repetitions) whether the
probability distribution over outputs $S\in\Phi_{m,n}$\ was consistent with
equation (\ref{per})---or in other words, with the prediction of quantum
mechanics.\footnote{To be clear, one would \textit{not} try to estimate
$\Pr\left[  S\right]  $\ for each of the exponentially-many possible outputs
$S$, since even for (say) $n=10,m=40$, that would require an impractical
amount of data-collection. \ Instead, one would simply verify that the
histogram of $\frac{\left\vert \operatorname*{Per}\left(  A_{S}\right)
\right\vert ^{2}}{s_{1}!\cdots s_{m}!}$\ for the $S$'s\ that \textit{were}
sampled was consistent with equation (\ref{per}).} \ Assuming our
complexity-theoretic conjectures, as $n$ and $m$ increased, those predictions
would rapidly get harder to reproduce by any simulation running on a classical computer.

In this way, one might hope to get \textquotedblleft experimental
evidence\textquotedblright\ against the \textit{Extended Church-Turing
Thesis}: i.e., the thesis that all physical processes can be simulated by a
classical computer with polynomial overhead. \ Furthermore, one might hope to
get such evidence more easily than by building a universal quantum computer.

Last year, four independent groups (based in Brisbane \cite{ebs:broome},
Oxford \cite{ebs:spring}, Vienna \cite{ebs:tillmann}, and Rome
\cite{ebs:crespi}) reported the first experiments more-or-less along the above
lines. \ In these experiments, $n$ (the number of photons) was generally
$3$,\footnote{Spring et al.\ \cite{ebs:spring}\ also managed to test $n=4$,
but for input states consisting of two modes with two photons each, rather
than four modes with one photon each.} while $m$ (the number of output modes)
was $5$ or $6$. \ The experiments directly confirmed, apparently for the first
time, the prediction of quantum mechanics that the amplitudes of $3$-photon
processes are given by permanents of $3\times3$\ complex matrices.

Obviously, these experiments do not yet provide any speedup over classical
computing, nor are their results surprising: at some level they
merely\ confirm quantum mechanics! \ But these are just the first steps. \ The
eventual goal would be to demonstrate \textsc{BosonSampling}\ with (say)
$n=20$ or $n=30$ photons: a regime where the quantum experiment probably
\textit{would} outperform its fastest classical simulation, if not by an
astronomical amount. \ In our view, this would be an exciting
proof-of-principle for quantum computation.

Scaling up \textsc{BosonSampling}\ to larger $n$ remains a nontrivial
experimental challenge. \ If it's possible at all, it will likely require
optical technologies (especially single-photon sources) much more reliable
than those that exist today. \ Indeed, we regard it as an open question
whether \textsc{BosonSampling}\ experiments \textit{can} be scaled to a
\textquotedblleft computationally interesting regime,\textquotedblright%
\ without the use of quantum fault-tolerance. \ And presumably, if one can
implement quantum fault-tolerance, then one might as well just skip
\textsc{BosonSampling}\ and build a universal quantum computer!

\section{The Claims of Gogolin et al.\label{CLAIMS}}

The above issues with \textsc{BosonSampling}---the lack of a known practical
motivation for it, the difficulties in scaling it up, etc.---are real and
well-known. \ We have tried to be clear about them from the outset. \ However,
in a recent preprint entitled \textquotedblleft Boson-Sampling in the light of
sample complexity,\textquotedblright\ Gogolin et al.\ \cite{gogolin} criticize
\textsc{BosonSampling}\ on different and much more theoretical grounds.
\ Namely, they claim that the output of even an \textit{ideal}
\textsc{BosonSampling}\ device would be \textquotedblleft operationally
indistinguishable\textquotedblright\ from the uniform distribution. \ Indeed,
they prove a theorem, which they interpret to mean that under
\textquotedblleft reasonable assumptions,\textquotedblright\ a classical
skeptic could never tell whether a claimed \textsc{BosonSampling}\ device\ was
simply outputting uniformly random noise.

Gogolin et al.\ add that \textquotedblleft it is important to note that our
findings do not contradict the results of [Aaronson and Arkhipov
\cite{aark}].\textquotedblright\ \ Yet despite this disclaimer, they strongly
imply that \cite{aark}\ overlooked an elementary point, one that severely
undermines the prospect of using \textsc{BosonSampling}\ to probe the Extended
Church-Turing Thesis.

In Sections \ref{SYMMETRIC}\ and \ref{INTRACT}, we will explain in detail why
Gogolin et al.\ are wrong. \ First, in Section \ref{SYMMETRIC}, we consider
their observation that so-called \textquotedblleft symmetric
algorithms\textquotedblright\ require exponentially many samples to
distinguish a Haar-random \textsc{BosonSampling}\ distribution $\mathcal{D}%
_{A}$ from the uniform distribution. \ We explain why their restriction to
\textquotedblleft symmetric algorithms\textquotedblright\ is absurd: if one
makes it, then countless other distributions become \textquotedblleft
indistinguishable from uniform,\textquotedblright\ even though they are
trivial to distinguish from uniform in reality!

Next, in Section \ref{INTRACT}, we consider Gogolin et al.'s \textquotedblleft
fallback position\textquotedblright: that, even if one allows non-symmetric
algorithms, distinguishing $\mathcal{D}_{A}$\ from the uniform distribution
could still be a hard \textit{computational} problem. \ We point out that we
made exactly the same observation in \cite{aark}---but that we also explained
in \cite{aark} why the asymptotic hardness of verification will \textit{not}
be the limiting factor, in practice, for interesting \textsc{BosonSampling}%
\ experiments (with, say, $n=30$\ photons) designed to probe the Extended
Church-Turing Thesis. \ (Unfortunately, Gogolin et al.\ never acknowledged or
engaged with this point.)

\section{Our Results\label{RESULTS}}

Even though we believe that the above completely suffices to answer Gogolin et
al., in Sections \ref{DEV}\ and \ref{DETECT}\ we go further, and address some
interesting technical questions raised by their work. \ In particular, once we
get over the confusion about \textquotedblleft symmetric
algorithms,\textquotedblright\ it's clear on numerical and heuristic grounds
that a generic \textsc{BosonSampling}\ distribution\ $\mathcal{D}_{A}$\ is
\textit{not} close to the uniform distribution. \ But can we rigorously
\textit{prove} that $\mathcal{D}_{A}$ is not close to uniform? \ (This, of
course, is necessary though not sufficient to prove that sampling from
$\mathcal{D}_{A}$\ is computationally intractable.) \ Also, is there a
polynomial-time classical algorithm to \textit{distinguish} $\mathcal{D}_{A}$
from the uniform distribution? \ What about from any efficiently-samplable
distribution? \ Finally, what can we say about \textsc{FermionSampling}%
\ (defined in terms of the determinant rather than the permanent), whose
statistical properties seem easier to understand?

Our results are as follows. \ In Section \ref{DEV}, we prove that a generic
\textsc{BosonSampling}\ distribution is \textit{not} close in variation
distance to the uniform distribution. \ We get this as a consequence of a
simple but nice fact, which could have independent applications to
\textsc{BosonSampling}: that, if $X$\ is an iid Gaussian matrix, then
$\left\vert \operatorname*{Per}\left(  X\right)  \right\vert ^{2}$\ is a
mixture of exponentially-distributed random variables. \ Then in Section
\ref{DETECT}, we describe a simple estimator $R^{\ast}$ (the squared product
of row-norms, scaled so that $\operatorname*{E}\left[  R^{\ast}\right]  =1$),
which we prove can distinguish a generic \textsc{BosonSampling}\ distribution
$\mathcal{D}_{A}$\ from the uniform distribution with constant bias and\ in
classical polynomial time. \ Let us state our result formally:

\begin{theorem}
\label{detect}Let $A\in\mathbb{C}^{m\times n}$ be a Haar-random
\textsc{BosonSampling}\ matrix with $m\geq n^{5.1}/\delta$. \ Let
$\mathcal{U}$\ be the uniform distribution over all experimental outcomes
$S\in\Phi_{m,n}$ (or over all \textquotedblleft
collision-free\textquotedblright\ outcomes, i.e., those with $s_{i}\in\left\{
0,1\right\}  $\ for all $i$), and let $\mathcal{D}_{A}$\ be the
\textsc{BosonSampling}\ distribution corresponding to $A$. \ There exists a
linear-time computable estimator $R^{\ast}$\ such that, for sufficiently large
$n$, and with probability $1-O\left(  \delta\right)  $\ over $A$, we have%
\[
\Pr_{S\sim\mathcal{D}_{A}}\left[  R^{\ast}\left(  A_{S}\right)  \geq1\right]
-\Pr_{S\sim\mathcal{U}}\left[  R^{\ast}\left(  A_{S}\right)  \geq1\right]
\geq\frac{1}{9}.
\]
In particular, this implies that, with $1-O\left(  \delta\right)
$\ probability, $\mathcal{D}_{A}$\ and $\mathcal{U}$\ have $\Omega\left(
1\right)  $\ variation distance.
\end{theorem}

To clarify, the estimator $R^{\ast}$\ does not distinguish $\mathcal{D}_{A}%
$\ from \textit{any} efficiently-samplable distribution; indeed, we show in
Section \ref{MOCKUP}\ that there are even natural \textquotedblleft
classical\textquotedblright\ models that produce the same statistics for
$R^{\ast}$ as correct \textsc{BosonSampling}. \ However, $R^{\ast}$ does
confirm that the output of a purported \textsc{BosonSampling}\ device has
nontrivial dependence on the beamsplitter settings, of a sort consistent with
its working correctly. \ So, this could be combined with other evidence to
build up a circumstantial case that a purported \textsc{BosonSampling}\ device
works, even with (say) $100$ or $1000$ photons.

Thus, in Appendix \ref{FIXED}, we study the broader question of
\textsc{BosonSampling}\ versus \textit{any} efficiently-samplable
distribution. \ We first observe that, for any \textit{fixed} $k$, it is easy
to construct an efficiently-samplable distribution that is
indistinguishable---unconditionally!---from a \textsc{BosonSampling}%
\ distribution $\mathcal{D}_{A}$ by any circuit of size at most $n^{k}$.
\ Indeed, this observation has nothing to do with \textsc{BosonSampling}: it
follows from a Chernoff bound, and holds for any target distribution
whatsoever. \ On the other hand, the \textquotedblleft
mockup\textquotedblright\ distribution thus constructed has only $O\left(
\log n\right)  $ entropy. \ So one could ask whether such a mockup
distribution exists that \textit{also} has large entropy. \ Here we report an
observation due to Brandao (personal communication): namely, that for every
$k$, a general theorem of Trevisan, Tulsiani, and Vadhan \cite{ttv} can be
used to construct an efficiently-samplable distribution that is
indistinguishable from a generic \textsc{BosonSampling}\ distribution
$\mathcal{D}_{A}$\ by circuits of size at most $n^{k}$, and that \textit{also}
has $n-O\left(  \log n\right)  $\ entropy. \ Of course, all of this leaves
open the crucial question of whether or not there is a \textit{single}
efficiently-samplable distribution that cannot be distinguished from
$\mathcal{D}_{A}$\ by any polynomial-time algorithm.

Next, in Appendix \ref{INITIAL}, we\ sketch an argument that the estimator
$R^{\ast}$\ works to distinguish a \textsc{BosonSampling}\ distribution from
uniform,\ given \textit{any} initial state (pure or mixed) with all photons
concentrated in the first $n\ll m$ modes, and which has a non-negligible
probability of a nonzero number of photons much less than $m$. \ In
particular, this implies that $R^{\ast}$ is \textquotedblleft
robust\textquotedblright: it still works even if a large fraction of photons
are randomly lost to the environment, and even if the inputs are (say)
coherent or Gaussian states rather than single-photon Fock states.

Finally, Appendix \ref{FERMION} presents some results about the related
problem of \textsc{FermionSampling}. \ \ In particular, we give a
self-contained proof that \textsc{FermionSampling}\ is solvable in classical
polynomial time. \ This was shown previously by Terhal and DiVincenzo
\cite{td:fermion}\ and by Knill \cite{knill:matchgate}\ (and was implicit in
work of Valiant \cite{valiant:qc}). \ However, our algorithm, which runs in
$O\left(  mn^{2}\right)  $ time, is both simpler and faster than any
previously-published \textsc{FermionSampling} algorithm, and seems like an
obvious choice for implementations. \ The existence of this
algorithm\ underscores that neither the \textquotedblleft
quantum\textquotedblright\ nature of \textsc{BosonSampling}, nor its
exponentially-large Hilbert space, nor its $n$-particle interference can
possibly suffice for computational hardness. \ This is why, contrary to the
claims of, e.g., Gard et al.\ \cite{gard}, we do not think it is possible to
explain convincingly why \textsc{BosonSampling}\ should be a hard problem
without using tools from computational complexity theory, as we did in
\cite{aark}.

In Appendix \ref{FERMION},\ we also reuse techniques from Section
\ref{DETECT}\ to understand the statistical properties of Haar-random
\textsc{FermionSampling}\ distributions. \ This turns out to be relatively
easy, owing to the fact---which we prove for completeness---that $\left\vert
\operatorname*{Det}\left(  X\right)  \right\vert ^{2}$\ converges at an
$O\left(  \log^{-3/2}n\right)  $\ rate to a lognormal random variable, given a
matrix $X\in\mathbb{C}^{n\times n}$\ of iid Gaussians. \ The convergence of
$\left\vert \operatorname*{Det}\left(  X\right)  \right\vert ^{2}$\ to
lognormal was previously shown by Girko \cite{girko}\ and by Costello and Vu
\cite{cv}, but for real $X$\ and without bounding the convergence rate. \ Note
that numerically, the pdfs for $\left\vert \operatorname*{Det}\left(
X\right)  \right\vert ^{2}$\ and $\left\vert \operatorname*{Per}\left(
X\right)  \right\vert ^{2}$ look nearly identical (see Figure \ref{pccfig}).
\ Thus, we conjecture that $\left\vert \operatorname*{Per}\left(  X\right)
\right\vert ^{2}$\ converges to lognormal as well; if true, this would give us
a much more detailed statistical understanding of Haar-random
\textsc{BosonSampling}\ distributions.

\section{Preliminaries\label{PRELIM}}

We use $\left[  n\right]  $\ to denote $\left\{  1,\ldots,n\right\}  $.
\ Given two probability distributions $\mathcal{D}_{1}=\left\{  p_{x}\right\}
_{x}$ and $\mathcal{D}_{2}=\left\{  q_{x}\right\}  _{x}$, the
\textit{variation distance}%
\[
\left\Vert \mathcal{D}_{1}-\mathcal{D}_{2}\right\Vert :=\frac{1}{2}\sum
_{x}\left\vert p_{x}-q_{x}\right\vert
\]
captures the maximum bias with which a sample from $\mathcal{D}_{1}$\ can be
distinguished from a sample from $\mathcal{D}_{2}$.

We already, in Section \ref{INTRO}, defined the \textsc{BosonSampling}%
\ problem and most of the notation we will use in discussing it. \ However,
one issue we need to get out of the way is that of multiple photons in the
same mode: something that, from our perspective, is mostly an inconvenience
that can be made irrelevant by taking sufficiently many modes. \ Formally,
call an experimental outcome $S=\left(  s_{1},\ldots,s_{m}\right)  \in
\Phi_{m,n}$ \textit{collision-free} if each $s_{i}$\ is either $0$ or $1$---so
that $A_{S}$\ is simply an $n\times n$ submatrix of $A$, and $\Pr
_{\mathcal{D}_{A}}\left[  S\right]  $\ is simply $\left\vert
\operatorname*{Per}\left(  A_{S}\right)  \right\vert ^{2}$. \ Also, let
$\Lambda_{m,n}\subseteq\Phi_{m,n}$\ be the set of all collision-free $S$.
\ Note that $\left\vert \Lambda_{m,n}\right\vert =\binom{m}{n}$, which means
that%
\begin{equation}
\left\vert \Lambda_{m,n}\right\vert \geq\left(  1-\frac{n^{2}}{m}\right)
\left\vert \Phi_{m,n}\right\vert . \label{cfree}%
\end{equation}
In this paper, we will typically assume that $m\gg n^{2}$ (or, for technical
reasons, even larger lower bounds on $m$), in which case (\ref{cfree}) tells
us that \textit{most} outcomes are collision-free. \ Moreover, in the case
that $A$ is Haar-random, the following result from \cite{aark}\ justifies
restricting our attention to the collision-free outcomes $S\in\Lambda_{m,n}$\ only:

\begin{theorem}
[\cite{aark}]\label{birthday}Let $A\in\mathbb{C}^{m\times n}$\ be a
Haar-random \textsc{BosonSampling}\ matrix. \ Then%
\[
\operatorname*{E}_{A}\left[  \Pr_{S\sim\mathcal{D}_{A}}\left[  S\notin%
\Lambda_{m,n}\right]  \right]  <\frac{2n^{2}}{m}.
\]

\end{theorem}

\section{Limitations of \textquotedblleft Symmetric
Algorithms\textquotedblright\label{SYMMETRIC}}

Suppose we want to verify that the output of a \textsc{BosonSampling}\ device
matches the predictions of quantum mechanics (that is, equation (\ref{per})).
\ Then given the matrix $A\in\mathbb{C}^{m\times n}$, our task can be
abstracted as that of designing a \textit{verification test}, $V_{A}%
:\Phi_{m,n}^{k}\rightarrow\left\{  0,1\right\}  $, that satisfies the
following two constraints:

\begin{itemize}
\item \textbf{Efficiency.} $\ V_{A}\left(  S_{1},\ldots,S_{k}\right)  $ can be
computed classically in time polynomial in $m$, $n$, and $k$.

\item \textbf{Usefulness.} $\ V_{A}\left(  S_{1},\ldots,S_{k}\right)  $ is
usually $1$ if the outcomes $S_{1},\ldots,S_{k}$\ are drawn from
$\mathcal{D}_{A}$, but usually $0$ if $S_{1},\ldots,S_{k}$ are generated in
various \textquotedblleft fake\textquotedblright\ ways (with the relevant
\textquotedblleft fake\textquotedblright\ ways depending on exactly what we
are trying to verify).
\end{itemize}

In this paper, we will typically assume two additional properties:

\begin{itemize}
\item \textbf{Uniformity.} \ The polynomial-time algorithm to compute
$V_{A}\left(  S_{1},\ldots,S_{k}\right)  $\ takes $A$ as part of its input,
rather than being a different algorithm for each $A$.

\item \textbf{Properness.} \ $V_{A}\left(  S_{1},\ldots,S_{k}\right)  $
distinguishes $\mathcal{D}_{A}$\ from the \textquotedblleft
fake\textquotedblright\ distributions, even if we consider $V_{A}$'s behavior
only in the case where no $S_{i}$\ contains collisions---i.e., where $S_{i}%
\in\Lambda_{m,n}$\ for all $i\in\left[  k\right]  $.
\end{itemize}

The motivation for the properness constraint is that our hardness results in
\cite{aark}\ used the collision-free outcomes $S\in\Lambda_{m,n}$ only, so one
might demand that one \textit{also} restrict oneself to $\Lambda_{m,n}$ only
when verifying that a hard problem is being solved.

Now, Gogolin et al.\ \cite{gogolin}\ insist on a further constraint on $V_{A}%
$, which they call \textquotedblleft symmetry.\textquotedblright

\begin{itemize}
\item \textbf{\textquotedblleft Symmetry.\textquotedblright} \ $V_{A}\left(
S_{1},\ldots,S_{k}\right)  $\ is a function only of the list of
\textit{multiplicities} of $S_{i}$'s within $\left(  S_{1},\ldots
,S_{k}\right)  $, and not of any other information about the $S_{i}$'s. \ In
other words, $V_{A}\left(  S_{1},\ldots,S_{k}\right)  $\ is invariant under
arbitrary permutations of the exponentially-large set $\Lambda_{m,n}$.
\end{itemize}

To illustrate, suppose that an alleged \textsc{BosonSampling}\ device with
$n=2$ photons and $m=3$ modes were run $k=3$ times, and suppose the
collision-free outcomes were:%
\[
\left\vert 1,1,0\right\rangle ,~\left\vert 1,1,0\right\rangle ,~\left\vert
0,1,1\right\rangle .
\]
Then a symmetric verifier would be told\ only that one outcome occurred twice
and that one outcome occurred once. \ From that information alone---together
with knowledge of $A\in\mathbb{C}^{m\times n}$---the verifier would have to
decide whether the device was sampling from the \textsc{BosonSampling}
distribution $\mathcal{D}_{A}$, or (the \textquotedblleft null
hypothesis\textquotedblright) whether it was just sampling from $\mathcal{U}$,
the uniform distribution over $S\in\Lambda_{m,n}$.

Most of Gogolin et al.'s paper \cite{gogolin}\ is devoted to proving that,
\textit{if} $A$ is drawn from the Haar measure, and $k$\ (the number of
experimental runs) is less than exponential in $n$, then with high probability
over $A$, the symmetric algorithm's task is information-theoretically
impossible. \ We believe their proof of this theorem to be correct. \ The
intuition is extremely simple: let $\tau\left(  A\right)  $\ be the expected
number of runs until \textit{any} outcome $S\in\Lambda_{m,n}$\ is observed
more than once. \ Then we will have $\tau\left(  A\right)  =2^{\Omega\left(
n\right)  }$, with overwhelming probability over the choice of $A$! \ This is
just because $\left\vert \Lambda_{m,n}\right\vert =\binom{m}{n}$\ is
exponentially large---and while $\mathcal{D}_{A}$\ is \textit{not} close in
variation distance to $\mathcal{U}$, neither is it concentrated on some tiny
subset of size $2^{o\left(  n\right)  }$. \ Thus, regardless of whether the
\textquotedblleft true\textquotedblright\ distribution is $\mathcal{D}_{A}%
$\ or $\mathcal{U}$, and notwithstanding the quadratic \textquotedblleft
speedup\textquotedblright\ obtained from the Birthday Paradox, after
$k=2^{o\left(  n\right)  }$\ runs, a symmetric algorithm is overwhelmingly
likely to have only the useless information, \textquotedblleft no sample
$S\in\Lambda_{m,n}$\ was observed more than once so far.\textquotedblright%
\ \ Or as Gogolin et al.\ put it:

\begin{quotation}
\noindent with probability exponentially close to one in the number of bosons,
no symmetric algorithm can distinguish the Boson-Sampling distribution from
the uniform one from fewer than exponentially many samples. \ This means that
the two distributions are operationally indistinguishable without detailed a
priori knowledge ... The realistic situation, at least as far as known
certification methods are concerned, much more closely resembles the black box
setting ... In this setting the certifier has no a priori knowledge about the
output distribution. \ It is hence reasonable to demand that its decision
should be independent of which particular samples he receives and only depend
on how often he receives them. \ That is to say, knowing nothing about the
probability distribution, the labels of the collected samples don't mean
anything to the certifier, hence they should not influence his decision
...\ [O]ur findings imply that in the black box setting distinguishing the
Boson-Sampling distribution from the uniform one requires exponentially many
samples ... Colloquially speaking, our results on this problem give rise to a
rather ironic situation: Instead of building a device that implements
Boson-Sampling, for example by means of a quantum optical experiment, one
could instead simply program a classical computer to efficiently sample from
the uniform distribution over [outputs] and claim that the device samples from
the post-selected Boson-Sampling distribution [for some unitary $U$]. \ If one
chooses $U$ from the Haar measure the chances of being caught cheating becomes
significantly large only after one was asked for exponentially many samples.
\ This implies that the findings of any experimental realisation of
Boson-Sampling have to be interpreted with great care, as far as the notion
\textquotedblleft quantum supremacy\textquotedblright\ [sic] is concerned.
\end{quotation}

Our response to these claims can be summed up in one sentence: \textit{there
is no reason whatsoever to restrict the verifier to symmetric algorithms only.
\ }The verifier's goal is to check whether the sampled distribution is a good
match to the ideal \textsc{BosonSampling}\ distribution $\mathcal{D}_{A}$.
\ Moreover, \textit{even under Gogolin et al.'s assumptions, }the verifier
knows the matrix $A$. \ And that makes sense: after all, $A$ is not secret,
but simply the input to the \textsc{BosonSampling}\ problem---just like some
particular positive integer is the input to \textsc{Factoring}, or some
particular graph is the input to \textsc{Hamilton Cycle}. \ So it seems
bizarre to throw away the information about the actual identities of the
observed outcomes $S\in\Lambda_{m,n}$, since the whole point is to compare
those outcomes to $\mathcal{D}_{A}$, given knowledge of $A$.

By analogy, supposing we were testing an algorithm for factoring a composite
number $N$ into its prime factors $p$ and $q$. \ Would anyone demand\ that our
verdict be independent of \textit{whether or not }$p\times q$\textit{ actually
equalled }$N$\textit{?} \ Would anyone say that our decision should depend
only on the shape of the histogram of probabilities for various output pairs
$\left(  p,q\right)  $, and be insensitive to the actual identities of the
$\left(  p,q\right)  $ pairs themselves? \ If not, then why impose such a
strange constraint on \textsc{BosonSampling}\ verifiers?

As a side note, suppose we decided, for some reason, that the verifier's
output had to be invariant under arbitrary relabelings of the \textit{input
and output modes}, though not necessarily of the entire set $\Lambda_{m,n}$
(call a verifier \textquotedblleft weakly symmetric\textquotedblright\ if it
has this property). \ Even then, we will show in Theorem \ref{weaksym}\ that
$\operatorname*{poly}\left(  m,n\right)  $ samples information-theoretically
suffice to distinguish $\mathcal{D}_{A}$\ from $\mathcal{U}$, with high
probability over $A$. \ The intuitive reason is that there are
\textquotedblleft only\textquotedblright\ $m!n!$\ possible relabelings of the
input and output modes, compared to $\binom{m}{n}!$\ relabelings of the
outcomes $S\in\Lambda_{m,n}$. \ So let $V_{A}$ be a verifier that errs on
input $\mathcal{U}$ with probability $\varepsilon\ll\frac{1}{m!n!}%
$---something we can easily achieve with $\operatorname*{poly}\left(
m,n\right)  $\ samples, using amplification.\ \ Then a weakly symmetric
verifier can simply run $V_{A}$\ for all $m!n!$\ possible mode-relabelings,
and check whether \textit{any} of them yield a good match between
$\mathcal{D}_{A}$\ and the experimental results.

In truth, though, there is no reason to restrict to weakly symmetric verifiers
either. \ Instead, the verifier should just get the raw list of experimental
outcomes,\ $S_{1},\ldots,S_{k}\in\Lambda_{m,n}$. \ In that case, Gogolin et
al.\ themselves state in their Theorem 3 that $O\left(  n^{3}\right)
$\ samples suffice to distinguish $\mathcal{D}_{A}$\ from $\mathcal{U}%
$\ information-theoretically, assuming the variation distance $\epsilon
=\left\Vert \mathcal{D}_{A}-\mathcal{U}\right\Vert $\ is a constant. \ This is
true, but much too weak: in fact $O\left(  1/\epsilon^{2}\right)  $\ samples
suffice to distinguish \textit{any} two probability distributions with
variation distance $\epsilon$. \ So the real issue is just to show that
$\epsilon=\Omega\left(  1\right)  $\ with high probability over $A$. \ That is
what we will do in Section \ref{DEV}.

\section{Intractability of Verification\label{INTRACT}}

Gogolin et al.'s second criticism is that, even if the verification algorithm
is given access to the input matrix $A$\ (which they strangely call
\textquotedblleft side information\textquotedblright), it \textit{still} can't
verify in polynomial time that a claimed \textsc{BosonSampling}\ device is
working directly. \ Indeed, this is the sole argument they offer for
restricting attention to \textquotedblleft symmetric\textquotedblright%
\ algorithms.\footnote{We still find the argument bizarre: if verifying a
\textsc{BosonSampling}\ device is computationally hard, then why make the
problem \textit{even harder}, by throwing away highly-relevant information?
\ The idea seems to be that, if non-symmetric algorithms need exponential
computation time (albeit very few samples), then at least no one can complain
that, by restricting to symmetric algorithms, we are prohibiting a known
polynomial-time algorithm. \ As we will explain, the trouble with this
argument is that, even \textit{assuming} we need $\exp\left(  n\right)  $
computation time, an algorithm that needs only $\operatorname*{poly}\left(
n\right)  $\ experimental samples is vastly preferable in practice to one that
needs $\exp\left(  n\right)  $\ samples.} \ As they write:

\begin{quotation}
\noindent due to the very fact that Boson-Sampling is believed to be hard,
efficient classical certification of Boson-Sampling devices seems to be out of
reach ... The complexity theoretic conjecture under which Boson-Sampling is a
hard sampling problem, namely that it is expected to be \#P hard to
approximate the permanent, implies that approximating the probabilities of the
individual outputs of a Boson-Sampling device is also computationally hard.
\ A classical certifier with limited computational power will hence have only
very limited knowledge about the ideal output distribution of a supposed
Boson-Sampling device ... it is certainly unreasonable to assume that the
[certifier] has full knowledge of the ideal Boson-Sampling distribution.
\ After all, it is the very point of Boson-Sampling that approximating the
probabilities of individual outcomes is a computationally hard problem ... Our
results indicate that even though, unquestionably, the Boson-Sampling
distribution has an intricate structure that makes sampling from it a
classically hard problem, this structure seems inaccessible by classical means.
\end{quotation}

While Gogolin et al.\ never mention this, we raised in the same point in
\cite[Section 1.3]{aark}:

\begin{quotation}
\noindent[U]nlike with \textsc{Factoring}, we do not believe there is any
$\mathsf{NP}$\ witness for \textsc{BosonSampling}. \ In other words, if $n$ is
large enough that a classical computer cannot solve \textsc{BosonSampling},
then $n$ is probably \textit{also} large enough that a classical computer
cannot even verify that a quantum computer is solving \textsc{BosonSampling} correctly.
\end{quotation}

And again, in \cite[Section 6.1]{aark}:

\begin{quotation}
\noindent Unlike with \textsc{Factoring}, we do not know of any
\textit{witness} for \textsc{BosonSampling} that a classical computer can
efficiently verify, much less a witness that a boson computer can produce.
\ This means that, when $n$ is very large (say, more than $100$), even if a
linear-optics device is correctly solving \textsc{BosonSampling}, there might
be no feasible way to prove this without presupposing the truth of the
physical laws being tested!
\end{quotation}

Note that Gogolin et al.\ do not offer any formal argument for why
\textsc{BosonSampling}\ verification is intractable, and neither did we. \ As
we'll see later, there are many different things one can \textit{mean} by
verification, and some of them give rise to fascinating technical questions,
on which we make some initial progress in this paper.\ \ In general, though,
even if we assume our hardness conjectures from \cite{aark}, we still lack a
satisfying picture of which types of \textsc{BosonSampling}\ verification can
and can't be done in classical polynomial time.

We did make the following observation in \cite[footnote 23]{aark}:

\begin{quotation}
\noindent[G]iven a\ matrix $X\in\mathbb{C}^{n\times n}$, there \textit{cannot}
in general be an $\mathsf{NP}$\ witness proving the value of
$\operatorname*{Per}\left(  X\right)  $, unless $\mathsf{P}^{\#\mathsf{P}%
}=\mathsf{P}^{\mathsf{NP}}$ and the polynomial hierarchy collapses. \ Nor,
under our conjectures, can there even be such a witness for \textit{most}
Gaussian matrices $X$.
\end{quotation}

However, one lesson of this paper is that the above does not suffice to show
that verification is hard. \ For it is possible that one can \textquotedblleft
verify\textquotedblright\ a \textsc{BosonSampling}\ distribution
$\mathcal{D}_{A}$\ (in the sense of distinguishing $\mathcal{D}_{A}$\ from a
large and interesting collection of \textquotedblleft null
hypotheses\textquotedblright), without having to verify the\ value of
$\operatorname*{Per}\left(  X\right)  $\ for any particular matrix $X$.

Having said all that, let's suppose, for the sake of argument, that it
\textit{is} computationally intractable to verify a\ claimed
\textsc{BosonSampling}\ device, for some reasonable definition of
\textquotedblleft verify.\textquotedblright\ \ If so, then what is our
response to that objection? \ Here, we'll simply quote the response we gave in
\cite[Section 1.3]{aark}:

\begin{quotation}
\noindent While [the intractability of verification] sounds discouraging, it
is not really an issue from the perspective of near-term experiments. \ For
the foreseeable future, $n$ being \textit{too large} is likely to be the least
of one's problems! \ If one could implement our experiment with (say)\ $20\leq
n\leq30$, then certainly a classical computer could verify the answers---but
at the same time, one would be getting direct evidence that a quantum computer
could efficiently solve an \textquotedblleft interestingly
difficult\textquotedblright\ problem, one for which the best-known classical
algorithms require many millions of operations.
\end{quotation}

And again, in \cite[Section 6.1]{aark}:

\begin{quotation}
\noindent\lbrack F]or experimental purposes, the most useful values of $n$ are
presumably those for which a classical computer has some difficulty computing
an $n\times n$\ permanent, but can nevertheless do so in order to confirm the results.
\end{quotation}

In other words: yes, there might be no way to verify a \textsc{BosonSampling}%
\ device when $n=100$. \ But $100$-photon \textsc{BosonSampling}\ is probably
neither experimentally feasible nor conceptually necessary \textit{anyway}!
\ The only known \textquotedblleft application\textquotedblright%
\ of\ \textsc{BosonSampling} is as a proof-of-principle: showing that a
quantum system can dramatically outperform its fastest classical simulation on
some task, under plausible complexity conjectures. \ And for that application,
all that matters is that there be \textit{some} range of $n$'s for which

\begin{enumerate}
\item[(1)] linear-optical devices and classical computers can both solve
\textsc{BosonSampling} in \textquotedblleft non-astronomical\textquotedblright%
\ amounts of time, but

\item[(2)] the linear-optical devices solve the problem noticeably faster.
\end{enumerate}

More concretely, as we pointed out in \cite{aark}, the fastest-known algorithm
for general $n\times n$ permanents, called \textit{Ryser's algorithm}, uses
about $2^{n+1}n$\ floating-point operations. \ If (say) $n=30$,\ then
$2^{n+1}n$ is about $64$\ billion, which is large but perfectly within the
capacity of today's computers. \ A classical computer therefore could feasibly
check that $\left\vert \operatorname*{Per}\left(  A_{S_{1}}\right)
\right\vert ^{2},\ldots,\left\vert \operatorname*{Per}\left(  A_{S_{k}%
}\right)  \right\vert ^{2}$\ satisfied the expected statistics, where
$S_{1},\ldots,S_{k}$\ were the outputs of an alleged $30$-photon
\textsc{BosonSampling}\ device. \ At the same time, the device would
presumably sample the $S_{i}$'s \textit{faster} than any known classical
method, for any reasonable definition of the word \textquotedblleft
faster.\textquotedblright\ \ If so, then \textsc{BosonSampling}\ would have
achieved its intended purpose: in our view, one would have done an experiment
that was harder than any previous experiment for a believer in the Extended
Church-Turing Thesis to explain.

\section{Deviation from Uniformity\label{DEV}}

Even if experimenters can live with $64$ billion floating-point operations,
they certainly \textit{can't} live with $64$ billion laboratory experiments!
\ So we still have the burden of showing that few experiments suffice to
distinguish a \textsc{BosonSampling}\ distribution $\mathcal{D}_{A}$ from
\textquotedblleft trivial\textquotedblright\ alternatives, and in particular
from the uniform distribution $\mathcal{U}$ over $\Lambda_{m,n}$. \ That is
what we will do in this section. \ In particular, we will prove that
$\left\Vert \mathcal{D}_{A}-\mathcal{U}\right\Vert =\Omega\left(  1\right)  $,
with $1-o\left(  1\right)  $ probability over a Haar-random $A\in
\mathbb{C}^{m\times n}$. \ Or equivalently, that the number of samples needed
to distinguish $\mathcal{D}_{A}$\ from $\mathcal{U}$%
\ (information-theoretically and with constant bias) is a constant,
independent of $n$.

Let $X=\left(  x_{ij}\right)  \in\mathbb{C}^{n\times n}$ be a matrix of iid
Gaussians, drawn from $\mathcal{N}=\mathcal{N}\left(  0,1\right)
_{\mathbb{C}}^{n\times n}$. \ Thus, $\mathcal{N}$\ has the pdf%
\[
f_{\mathcal{N}}\left(  X\right)  =\frac{1}{\pi^{n^{2}}}\exp\left(  -%
{\textstyle\sum\limits_{ij}}
\left\vert x_{ij}\right\vert ^{2}\right)  .
\]
Also, let%
\[
P=P\left(  X\right)  =\frac{\left\vert \operatorname*{Per}\left(  X\right)
\right\vert ^{2}}{n!},
\]
so that $\operatorname{E}\left[  P\right]  =1$, and let $f_{P}$\ be the pdf of
$P$. \ (We will typically suppress dependence on $n$.) \ Our goal, in this
section, is to understand some basic facts about the function $f_{P}$,
ignoring issues of computational complexity, and then use those facts to prove
that $\mathcal{D}_{A}$ is not close to the uniform distribution. \ For a plot
of $f_{P}$\ and the corresponding pdf for the determinant\ in the case $n=6$
(reproduced from \cite{aark}), see Figure \ref{pccfig}.

\definecolor{DarkBlue}{rgb}{0,0,1} \definecolor{DarkPink}{rgb}{1,0,0}
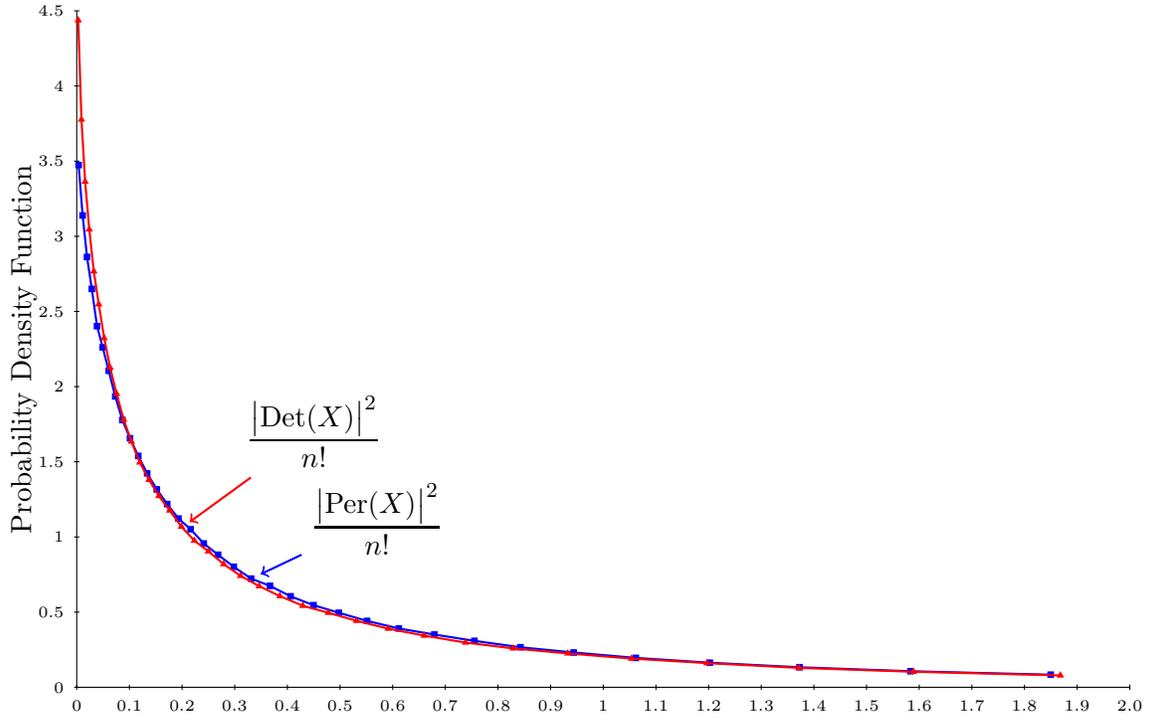
\begin{figure}[ptb]
\centering
\par
\begin{tikzpicture}[x=7cm, y=2cm]
\draw (0,0) -- (2,0) node[right] {};
\draw (0,0) -- (0,4.5) node[above] {};
\foreach \x in
{0,0.1,0.2,0.3,0.4,0.5,0.6,0.7,0.8,0.9,1,1.1,1.2,1.3,1.4,1.5,1.6,1.7,1.8,1.9,2.0}
\draw (\x ,1pt) -- (\x ,-1pt) node[anchor=north] {\tiny $\x$};
\foreach \y in {0, 0.5, ..., 4.5}
\draw (1pt,\y) -- (-1pt, \y) node[anchor=east] {\tiny $\y$};
\node[rotate=90] at (-0.1,2.25) {Probability Density Function};
\draw[color=DarkBlue, thick, mark=square*, mark size=0.85pt] plot
file {Perm_plot_points.table};
\draw[color=DarkPink, thick, mark = triangle, mark size = 0.85pt]
plot file {Det_plot_points.table};
\node (Det) at (0.45,1.7) {
$\displaystyle\frac{\bigl|\operatorname{Det}(X)\bigr|^2}{n!}$};
\node (DetPoint) at (0.195,1.05) {};
\draw[->, thick, color=DarkPink] (Det) -- (DetPoint);
\node (Per) at (0.57,1.12) {
$\displaystyle\frac{\bigl|\operatorname{Per}(X)\bigr|^2}{n!}$};
\node (PerPoint) at (0.33,0.72) {};
\draw[->, thick, color=DarkBlue] (Per) -- (PerPoint);
\end{tikzpicture}
\caption{Probability density functions of the random variables $D=\left\vert
\operatorname*{Det}\left(  X\right)  \right\vert ^{2}/n!$\ and $P=\left\vert
\operatorname*{Per}\left(  X\right)  \right\vert ^{2}/n!$, where
$X\sim\mathcal{N}\left(  0,1\right)  _{\mathbb{C}}^{n\times n}$ is a complex
Gaussian random matrix, in the case $n=6$ (reproduced from \cite{aark}).
\ Note that $\operatorname*{E}\left[  D\right]  =\operatorname*{E}\left[
P\right]  =1$. \ As $n$ increases, the bends on the left become steeper.}%
\label{pccfig}%
\end{figure}

The first step is to give a characterization of $f_{P}$\ that, while easy to
prove (it involves considering only the topmost row of the matrix $X$), will
provide a surprising amount of leverage.

\begin{lemma}
\label{expmix}There exists a random variable $c>0$ such that%
\[
f_{P}\left(  x\right)  =\operatorname{E}_{c}\left[  ce^{-cx}\right]  .
\]
In other words, $P\left(  X\right)  $\ is a (possibly continuous) mixture of
exponentially-distributed random variables.
\end{lemma}

\begin{proof}
Given $X=\left(  x_{ij}\right)  \in\mathbb{C}^{n\times n}$, let $X_{1}%
,\ldots,X_{n}$\ be the bottom $\left(  n-1\right)  \times\left(  n-1\right)
$\ minors. \ Then by definition,%
\[
\operatorname*{Per}\left(  X\right)  =x_{11}\operatorname*{Per}\left(
X_{1}\right)  +\cdots+x_{1n}\operatorname*{Per}\left(  X_{n}\right)  .
\]
Now, $\operatorname*{Per}\left(  X_{1}\right)  ,\ldots,\operatorname*{Per}%
\left(  X_{n}\right)  $ have some (correlated) probability measure, call it
$\mathcal{C}$. \ Then we can think of $\operatorname*{Per}\left(  X\right)  $
as simply $c_{1}x_{1}+\cdots+c_{n}x_{n}$, where $x_{1},\ldots,x_{n}$\ are
independent $\mathcal{N}\left(  0,1\right)  _{\mathbb{C}}$\ Gaussians, and
$\left(  c_{1},\ldots,c_{n}\right)  $\ is drawn from $\mathcal{C}$. \ This, in
turn, is a complex Gaussian with mean $0$ and variance $\left\vert
c_{1}\right\vert ^{2}+\cdots+\left\vert c_{n}\right\vert ^{2}$. \ Therefore
the pdf of\ $\operatorname*{Per}\left(  X\right)  $\ must be a convex
combination of complex Gaussians, each with mean $0$ (but with different
variances). \ So the pdf of $\left\vert \operatorname*{Per}\left(  X\right)
\right\vert ^{2}$\ is a convex combination of absolute squares of complex
Gaussians, which are exponentially-distributed random variables.
\end{proof}

As a side note, since $1/c$ is the expectation of the random variable with pdf
$ce^{-cx}$, we have%
\[
\operatorname{E}\left[  \frac{1}{c}\right]  =\operatorname{E}_{c}\left[
\int_{0}^{\infty}ce^{-cx}xdx\right]  =\int_{0}^{\infty}f_{P}\left(  x\right)
xdx=1.
\]

To illustrate the usefulness of Lemma \ref{expmix}, we now apply it to deduce
various interesting facts about $f_{P}$, explaining what were left in
\cite{aark}\ as empirical observations. \ The following theorem is not needed
for later results in this section, but is included for completeness.

\begin{theorem}
\label{grabbag}$f_{P}$\ is monotonically decreasing, finite, positive, and
smooth, except that it might diverge at $x=0$.
\end{theorem}

\begin{proof}
For monotonically decreasing and positive, it suffices to note that $f_{P}%
$\ is a convex combination of functions with those properties. \ For finite,%
\[
f_{P}\left(  x\right)  =\operatorname{E}_{c}\left[  ce^{-cx}\right]  \leq
\sup_{c>0}ce^{-cx}\leq\frac{1}{ex}%
\]
for $x>0$. \ Likewise for continuous,%
\begin{align*}
\frac{df_{P}\left(  x\right)  }{dx}  &  =\lim_{\varepsilon\rightarrow0}%
\frac{f_{P}\left(  x+\varepsilon\right)  -f_{P}\left(  x-\varepsilon\right)
}{2\varepsilon}\\
&  =\lim_{\varepsilon\rightarrow0}\frac{\operatorname{E}_{c}\left[
ce^{-c\left(  x+\varepsilon\right)  }-ce^{-c\left(  x-\varepsilon\right)
}\right]  }{2\varepsilon}\\
&  =\lim_{\varepsilon\rightarrow0}\operatorname{E}_{c}\left[  \frac
{ce^{-cx}\left(  e^{-c\varepsilon}-e^{c\varepsilon}\right)  }{2\varepsilon
}\right] \\
&  =\operatorname{E}_{c}\left[  -c^{2}e^{-cx}\right]
\end{align*}
which is finite for $x>0$, and a similar calculation can be repeated for the
higher derivatives.
\end{proof}

We now use Lemma \ref{expmix}\ to show, in two senses, that $f_{P}$ is far
from the point distribution concentrated on $P=1$.

\begin{lemma}
\label{313}We have%
\[
\frac{1}{2}\operatorname{E}\left[  \left\vert P-1\right\vert \right]
>0.313,~~~~\Pr\left[  \left\vert P-1\right\vert \geq\frac{1}{2}\right]
>0.615.
\]

\end{lemma}

\begin{proof}
For the first part: since $f_{P}$\ is a mixture of pdf's of the form
$ce^{-cx}$, by convexity it suffices to lower-bound%
\[
\inf_{c>0}\frac{1}{2}\int_{0}^{\infty}ce^{-cx}\left\vert x-1\right\vert dx.
\]
The expression inside the $\inf$\ can be evaluated as%
\[
\frac{1}{2}-\frac{1}{2c}+\frac{1}{ce^{c}},
\]
for which a minimum of $\sim0.313$\ at $c\approx1.678$\ can be obtained numerically.

For the second part: again appealing to convexity, it suffices to lower-bound%
\[
\inf_{c>0}\left(  1-\int_{1/2}^{3/2}ce^{-cx}dx\right)  ,
\]
for which a minimum of $1-\frac{2\sqrt{3}}{9}>0.615$\ is obtained at $c=\ln3$.
\end{proof}

Next, let $\mathcal{H}$\ be the distribution over $X\in\mathbb{C}^{n\times n}%
$\ defined by the pdf%
\[
f_{\mathcal{H}}\left(  X\right)  =f_{\mathcal{N}}\left(  X\right)  P\left(
X\right)  .
\]
In other words, $\mathcal{H}$\ is the distribution over $n\times n$\ matrices
obtained by starting from the Gaussian distribution $\mathcal{N}$, then
rescaling all the probabilities by $\left\vert \operatorname*{Per}\left(
X\right)  \right\vert ^{2}$. \ Then as an immediate corollary of Lemma
\ref{313}, we find that%
\begin{equation}
\left\Vert \mathcal{H}-\mathcal{N}\right\Vert =\frac{1}{2}\int_{X}\left\vert
f_{\mathcal{H}}\left(  X\right)  -f_{\mathcal{N}}\left(  X\right)  \right\vert
dX=\frac{1}{2}\operatorname{E}\left[  \left\vert P-1\right\vert \right]
>0.313. \label{line313}%
\end{equation}

Now let $A\in\mathbb{C}^{m\times n}$\ be a Haar-random column-orthonormal
matrix, and recall that $A_{S}$\ is the $n\times n$\ submatrix of $A$
corresponding to experimental outcome $S\in\Lambda_{m,n}$. \ Let
$\mathcal{H}^{\prime}$\ be the distribution over $n\times n$\ matrices
$X=\sqrt{m}A_{S}$, where $S$\ was drawn from the \textsc{BosonSampling}%
\ distribution $\mathcal{D}_{A}$ (averaged over all $A$). \ Likewise, let
$\mathcal{N}^{\prime}$\ be\ the distribution over $X=\sqrt{m}A_{S}$\ where
$S$\ was drawn from the uniform\ distribution $\mathcal{U}$. \ (Note that by
symmetry, we could have just as well set $X=\sqrt{m}A_{S}$\ for any
\textit{fixed} $S\in\Lambda_{m,n}$---say, the lexicographically first $S$.)

Then based on the results in \cite{aark}, we might guess that $\mathcal{N}%
^{\prime}\approx\mathcal{N}$ and $\mathcal{H}^{\prime}\approx\mathcal{H}$:
that is, a random $n\times n$\ submatrix of a Haar-random $A$\ should look
close to an iid Gaussian matrix, while a random submatrix whose probability
was scaled by $\left\vert \operatorname*{Per}\right\vert ^{2}$ should look
close to a $\left\vert \operatorname*{Per}\right\vert ^{2}$-scaled Gaussian
matrix. \ Thus, line (\ref{line313}) strongly suggests that%
\[
\left\Vert \mathcal{H}^{\prime}-\mathcal{N}^{\prime}\right\Vert \approx
\left\Vert \mathcal{H}-\mathcal{N}\right\Vert =\Omega\left(  1\right)  .
\]
Or in words: it should be easy to tell (information-theoretically and with
constant bias) whether an $n\times n$\ submatrix $A_{S}$\ was uniformly
sampled or \textquotedblleft BosonSampled\textquotedblright\ from among
the\ $\binom{m}{n}$\ submatrices of a Haar-random $A$, \textit{just by
examining }$A_{S}$\textit{\ itself} (and not even knowing $S$ or the rest of
$A$). \ Intuitively, this is because a BosonSampled $A_{S}$\ will tend to
\textquotedblleft stick out\textquotedblright\ by having an unusually large
$\left\vert \operatorname*{Per}\left(  A_{S}\right)  \right\vert ^{2}$.

As a consequence, one also expects that, with high probability over $A$, the
\textsc{BosonSampling}\ distribution $\mathcal{D}_{A}$\ should have
$\Omega\left(  1\right)  $ variation distance\ from the uniform distribution
$\mathcal{U}$. \ For once we decide whether the submatrix $A_{S}$ was drawn
from $\mathcal{H}^{\prime}$\ or from $\mathcal{N}^{\prime}$, that should then
tell us whether $S$\ itself was drawn from $\mathcal{D}_{A}$\ or $\mathcal{U}$.

Unfortunately, there are two technical difficulties in formalizing the above.
\ The first is that, if $f_{P}$\ (i.e., the pdf of $\left\vert
\operatorname*{Per}\left(  X\right)  \right\vert ^{2}/n!$) were
\textit{extremely} heavy-tailed---and we can't currently prove that it
isn't---then $\mathcal{H}^{\prime}$\ wouldn't need to be close to
$\mathcal{H}$\ in variation distance. \ Intuitively, the rare matrices $X$
such that $f_{\mathcal{N}^{\prime}}\left(  X\right)  $\ was far from
$f_{\mathcal{N}}\left(  X\right)  $\ could \textquotedblleft acquire an
outsized importance\textquotedblright\ for the scaled distributions
$\mathcal{H}^{\prime}$\ and $\mathcal{H}$, if such $X$'s had enormous
permanents. \ The second difficulty is that, if we want to show that
$\mathcal{D}_{A}$\ is far from $\mathcal{U}$\ with $1-o\left(  1\right)
$\ probability over $A$ (rather than merely $\Omega\left(  1\right)
$\ probability), then we need to say something about the lack of strong
correlations between \textit{different} submatrices $A_{S},A_{S^{\prime}}$\ of
the same $A$. \ For otherwise, it might be that all the $A_{S}$'s that caused
$\mathcal{H}^{\prime}$\ and $\mathcal{N}^{\prime}$\ to have constant variation
distance from each other, were concentrated within (say) $50\%$\ of the $A$'s.

Fortunately, we can overcome these difficulties, with help from a result
proved in \cite{aark}.

\begin{theorem}
[\cite{aark}]\label{trunc}Let $m\geq\frac{n^{5}}{\delta}\log^{2}\frac
{n}{\delta}$\ for any $\delta>0$. \ Then $\left\Vert \mathcal{N}^{\prime
}-\mathcal{N}\right\Vert =O\left(  \delta\right)  $.
\end{theorem}

As discussed in \cite{aark}, we are confident that the $n^{5}$\ in Theorem
\ref{trunc}\ is purely an artifact of the proof, and that it can be improved
to $n^{2}$. \ (Indeed, Jiang \cite{jiang}\ \textit{did} prove an analogue of
Theorem \ref{trunc} assuming only $m\gg n^{2}$, except for real orthogonal
matrices rather than unitary matrices, and without the explicit dependence on
$\delta$.)

In any case, by combining Theorem \ref{trunc}\ with the second part of Lemma
\ref{313}, we immediately obtain the following.

\begin{corollary}
\label{615delta}Assume $m\geq\frac{n^{5}}{\delta}\log^{2}\frac{n}{\delta}$.
\ Then%
\[
\Pr_{X\sim\mathcal{N}^{\prime}}\left[  \left\vert P\left(  X\right)
-1\right\vert \geq\frac{1}{2}\right]  \geq0.615-O\left(  \delta\right)  .
\]

\end{corollary}

Given a \textsc{BosonSampling} matrix $A\in\mathbb{C}^{m\times n}$, let
$\mathcal{B}_{A}$\ be the (discrete) distribution over matrices $X\in
\mathbb{C}^{n\times n}$\ obtained by first drawing $S$\ from $\mathcal{D}_{A}%
$, and then setting $X:=\sqrt{m}A_{S}$. \ Then next we need a lemma that
upper-bounds the probability, over $A$, that $\mathcal{B}_{A}$\ looks very
different from $\mathcal{H}^{\prime}$\ with respect to a certain statistical test.

\begin{lemma}
\label{walem}Let $A\in\mathbb{C}^{m\times n}$\ be a Haar-random
\textsc{BosonSampling} matrix with $m\geq n^{5.1}/\delta$. \ Also, let%
\[
W_{A}=\Pr_{S\in\Lambda_{m,n}}\left[  \left\vert m^{n}P\left(  A_{S}\right)
-1\right\vert \geq\frac{1}{2}\right]
\]
(recalling that $P\left(  X\right)  =\left\vert \operatorname*{Per}\left(
X\right)  \right\vert ^{2}/n!$). \ Then%
\[
\Pr_{A}\left[  W_{A}\leq\frac{1}{2}\right]  =O\left(  \frac{1}{n}\right)  .
\]
Or rewriting, we have%
\[
\Pr_{S\in\Lambda_{m,n}}\left[  \left\vert \frac{m^{n}}{n!}\left\vert
\operatorname*{Per}\left(  A_{S}\right)  \right\vert ^{2}-1\right\vert
\geq\frac{1}{2}\right]  \geq\frac{1}{2}%
\]
with probability $1-O\left(  \delta\right)  $\ over $A$.
\end{lemma}

\begin{proof}
For all $S\in\Lambda_{m,n}$, let $w_{S}:=1$\ if%
\[
\left\vert m^{n}P\left(  A_{S}\right)  -1\right\vert \geq\frac{1}{2}%
\]
and $w_{S}:=0$\ otherwise.\ \ Then clearly%
\[
W_{A}=\frac{1}{\left\vert \Lambda_{m,n}\right\vert }\sum_{S\in\Lambda_{m,n}%
}w_{S}.
\]
Hence%
\begin{align*}
\operatorname*{E}_{A}\left[  W_{A}\right]   &  =\frac{1}{\left\vert
\Lambda_{m,n}\right\vert }\sum_{S\in\Lambda_{m,n}}\operatorname*{E}_{A}\left[
w_{S}\right] \\
&  =\frac{1}{\left\vert \Lambda_{m,n}\right\vert }\sum_{S\in\Lambda_{m,n}}%
\Pr_{X\sim\mathcal{N}^{\prime}}\left[  \left\vert P\left(  X\right)
-1\right\vert \geq\frac{1}{2}\right] \\
&  \geq0.615-O\left(  \delta\right)  ,
\end{align*}
where the second line uses the definition of $\mathcal{N}^{\prime}$\ and the
third line uses Corollary \ref{615delta}.

Now consider \textit{two} experimental outcomes, $S,T\in\Lambda_{m,n}$.
\ Notice that, if $S$\ and $T$ are disjoint ($S\cap T=\varnothing$), then
$A_{S\cup T}$\ is simply a $2n\times n$\ submatrix of $A$. \ So, if we think
of \thinspace$A$ as an $m\times n$\ submatrix of a Haar-random $m\times
2n$\ matrix $A^{\prime}$, then $A_{S\cup T}$\ is a submatrix of a $2n\times
2n$\ submatrix of $A^{\prime}$. \ But this means that, if we set $n^{\prime
}:=2n$\ (which has no effect on the asymptotics), then we can apply Theorem
\ref{trunc}\ to $A_{S\cup T}$\ exactly as if it were an $n\times
n$\ submatrix. \ So in particular, $A_{S\cup T}$\ will be $O\left(
\delta\right)  $-close in variation distance to a $2n\times n$\ matrix of iid
Gaussians with mean $0$ and variance $1/m$---or in other words, to two
independent samples from $\mathcal{N}^{\prime}$.

We can use the above considerations to upper-bound the variance of $W_{A}$:%
\begin{align*}
\operatorname*{Var}_{A}\left[  W_{A}\right]   &  =\frac{1}{\left\vert
\Lambda_{m,n}\right\vert ^{2}}\sum_{S,T\in\Lambda_{m,n}}\left(
\operatorname*{E}_{A}\left[  w_{S}w_{T}\right]  -\operatorname*{E}_{A}\left[
w_{S}\right]  \operatorname*{E}_{A}\left[  w_{T}\right]  \right) \\
&  \leq O\left(  \delta\right)  +\frac{1}{\left\vert \Lambda_{m,n}\right\vert
^{2}}\sum_{S,T\in\Lambda_{m,n}:S\cap T\neq\varnothing}1\\
&  \leq\frac{n^{2}}{m}+O\left(  \delta\right)  .
\end{align*}
Here we used the facts that $\operatorname*{E}_{A}\left[  w_{S}w_{T}\right]
\leq1$, and that $\Pr_{S,T\in\Lambda_{m,n}}\left[  S\cap T\neq\varnothing
\right]  \leq n^{2}/m$\ by the union bound.

Combining and using Chebyshev's inequality,%
\begin{align*}
\Pr_{A}\left[  W_{A}<\frac{1}{2}\right]   &  <\frac{\operatorname*{Var}%
_{A}\left[  W_{A}\right]  }{\left(  \operatorname*{E}_{A}\left[  W_{A}\right]
-1/2\right)  ^{2}}\\
&  \leq\frac{n^{2}/m+O\left(  \delta\right)  }{\left(  0.615-O\left(
\delta\right)  -1/2\right)  ^{2}}\\
&  =O\left(  \frac{n^{2}}{m}+\delta\right) \\
&  =O\left(  \delta\right)  ,
\end{align*}
where the last line used that $m\geq n^{5.1}/\delta$.
\end{proof}

Once again, we are confident that the condition $m\geq n^{5.1}/\delta$\ in
Lemma \ref{walem} is not tight; indeed, one should be able to get something
whenever $m\geq n^{2}/\delta$.

The final ingredient in the proof is a simple fact about variation distance.

\begin{lemma}
\label{var4}Let $\mathcal{U}$\ be the uniform distribution over a finite set
$X$, and let $\mathcal{D}=\left\{  p_{x}\right\}  _{x\in X}$\ be some other
distribution. \ Also, let $Z\subseteq X$\ satisfy $\left\vert Z\right\vert
\geq\left(  1-\alpha\right)  \left\vert X\right\vert $, and let $M\in\left[
\left(  1-\beta\right)  \left\vert X\right\vert ,\left\vert X\right\vert
\right]  $. \ Then%
\[
\left\Vert \mathcal{D}-\mathcal{U}\right\Vert \geq\frac{1-\alpha}{4}\Pr_{x\in
Z}\left[  \left\vert Mp_{x}-1\right\vert \geq\frac{1}{2}\right]  -\frac{\beta
}{2-2\beta}.
\]

\end{lemma}

\begin{proof}
We have%
\begin{align*}
\left\Vert \mathcal{D}-\mathcal{U}\right\Vert  &  =\frac{1}{2}\sum_{x\in
X}\left\vert p_{x}-\frac{1}{\left\vert X\right\vert }\right\vert \\
&  \geq\frac{1}{2}\sum_{x\in Z}\left\vert p_{x}-\frac{1}{\left\vert
X\right\vert }\right\vert \\
&  \geq\frac{1}{2}\sum_{x\in Z}\left(  \left\vert p_{x}-\frac{1}{M}\right\vert
-\left(  \frac{1}{M}-\frac{1}{\left\vert X\right\vert }\right)  \right) \\
&  =\left(  \frac{1}{2M}\sum_{x\in Z}\left\vert Mp_{x}-1\right\vert \right)
-\frac{\beta}{1-\beta}\left(  \frac{\left\vert Z\right\vert }{2\left\vert
X\right\vert }\right) \\
&  \geq\frac{\left\vert Z\right\vert }{4M}\Pr_{x\in Z}\left[  \left\vert
Mp_{x}-1\right\vert \geq\frac{1}{2}\right]  -\frac{\beta}{1-\beta}\left(
\frac{\left\vert Z\right\vert }{2\left\vert X\right\vert }\right) \\
&  \geq\frac{1-\alpha}{4}\Pr_{x\in Z}\left[  \left\vert Mp_{x}-1\right\vert
\geq\frac{1}{2}\right]  -\frac{\beta}{2-2\beta}%
\end{align*}

\end{proof}

Combining Lemmas \ref{walem}\ and \ref{var4}\ now yields the main result of
the section.

\begin{theorem}
\label{vardist}Let $A\in\mathbb{C}^{m\times n}$ be a Haar-random
\textsc{BosonSampling}\ matrix with $m\geq n^{5.1}/\delta$. \ Let
$\mathcal{U}$\ be the uniform distribution over $\Lambda_{m,n}$, and let
$\mathcal{D}_{A}$\ be the \textsc{BosonSampling}\ distribution corresponding
to $A$. \ Then for sufficiently large $n$ and with probability $1-O\left(
\delta\right)  $\ over $A$,%
\[
\left\Vert \mathcal{D}_{A}-\mathcal{U}\right\Vert \geq\frac{1}{9}.
\]

\end{theorem}

\begin{proof}
In Lemma \ref{var4}, set%
\[
X=\Phi_{m,n},~~~~Z=\Lambda_{m,n},~~~~M=\frac{m^{n}}{n!}.
\]
Then $\left\vert X\right\vert =\binom{m+n-1}{n}$ and $\left\vert Z\right\vert
=\binom{m}{n}$, which means that we can set $\alpha:=n^{2}/m$\ and
$\beta:=n^{2}/m$. \ So using Lemma \ref{var4} and then Lemma \ref{walem}, we
find that, with probability $1-O\left(  \delta\right)  $\ over $A$:%
\begin{align*}
\left\Vert \mathcal{D}_{A}-\mathcal{U}\right\Vert  &  \geq\frac{1-\alpha}%
{4}\Pr_{x\in Z}\left[  \left\vert Mp_{x}-1\right\vert \geq\frac{1}{2}\right]
-\frac{\beta}{2-2\beta}\\
&  =\frac{1-n^{2}/m}{4}\Pr_{S\in\Lambda_{m,n}}\left[  \left\vert \frac{m^{n}%
}{n!}\left\vert \operatorname*{Per}\left(  A_{S}\right)  \right\vert
^{2}-1\right\vert \geq\frac{1}{2}\right]  -\frac{n^{2}/m}{2-2n^{2}/m}\\
&  \geq\frac{1-n^{2}/m}{8}-\frac{n^{2}/m}{2-2n^{2}/m}\\
&  \geq\frac{1}{8}-O\left(  \frac{n^{2}}{m}\right) \\
&  \geq\frac{1}{9}%
\end{align*}
for sufficiently large $n$.
\end{proof}

\subsection{Weakly-Symmetric Verifiers\label{WEAKLY}}

In this subsection, we use Theorem \ref{vardist}\ to justify a claim made in
Section \ref{SYMMETRIC}: namely, that it's possible to distinguish a generic
\textsc{BosonSampling}\ distribution $\mathcal{D}_{A}$ from the uniform
distribution $\mathcal{U}$, even if we restrict ourselves
to\ \textquotedblleft weakly-symmetric\textquotedblright\ verifiers (which
don't know the labels of the input and output modes).

The first step is to observe the following corollary of Theorem \ref{vardist}.

\begin{corollary}
\label{ampcor}Let $A\in\mathbb{C}^{m\times n}$ be a Haar-random
\textsc{BosonSampling}\ matrix with $m\geq n^{5.1}/\delta$. \ Then there
exists a proper verifier $V_{A}$\ (not necessarily computationally efficient)
and a constant $c>1$\ such that, with probability $1-O\left(  \delta\right)
$\ over $A$, the following holds:%
\begin{align}
\Pr_{S_{1},\ldots,S_{k}\sim\mathcal{D}_{A}}\left[  V_{A}\left(  S_{1}%
,\ldots,S_{k}\right)  \text{ accepts}\right]   &  \geq1-\frac{1}{c^{k}%
},\label{ck1}\\
\Pr_{S_{1},\ldots,S_{k}\sim\mathcal{U}}\left[  V_{A}\left(  S_{1},\ldots
,S_{k}\right)  \text{ accepts}\right]   &  \leq\frac{1}{c^{k}}. \label{ck2}%
\end{align}
Moreover, if (\ref{ck1})\ and (\ref{ck2}) hold for a given $A$, then they also
hold for any $P_{\sigma}AP_{\tau}$, where $P_{\sigma}$\ and $P_{\tau}$\ are
the permutation matrices corresponding to $\sigma\in S_{m}$\ and $\tau\in
S_{n}$\ respectively.
\end{corollary}

\begin{proof}
The inequalities (\ref{ck1})\ and (\ref{ck2})\ follow immediately from Theorem
\ref{vardist},\ together with a basic \textquotedblleft
amplification\textquotedblright\ property of variation distance: namely that,
by a Chernoff bound,%
\[
\left\Vert \mathcal{D}_{1}^{\otimes k}-\mathcal{D}_{2}^{\otimes k}\right\Vert
>1-\frac{1}{\exp\left(  \Omega\left(  k\cdot\left\Vert \mathcal{D}%
_{1}-\mathcal{D}_{2}\right\Vert _{1}^{2}\right)  \right)  }.
\]
If one wishes to be more explicit, the verifier $V_{A}\left(  S_{1}%
,\ldots,S_{k}\right)  $\ should accept if and only if%
\[%
{\displaystyle\prod\limits_{i=1}^{k}}
\left\vert \operatorname*{Per}\left(  A_{S_{i}}\right)  \right\vert ^{2}%
\geq\left(  \frac{n!}{m^{n}}\right)  ^{k}.
\]
For the last part, simply observe that, if we take $V_{A}$\ as above, then
$\Pr_{S_{1},\ldots,S_{k}\sim\mathcal{D}_{A}}\left[  V_{A}\text{ accepts}%
\right]  $\ and\ $\Pr_{S_{1},\ldots,S_{k}\sim\mathcal{U}}\left[  V_{A}\text{
accepts}\right]  $\ are both completely unaffected by permutations of the rows
and columns of $A$.
\end{proof}

Using Corollary \ref{ampcor},\ we can easily construct a weakly-symmetric verifier.

\begin{theorem}
\label{weaksym}Let $A\in\mathbb{C}^{m\times n}$ be a Haar-random
\textsc{BosonSampling}\ matrix with $m\geq n^{5.1}/\delta$, and let
$k=O\left(  \log\frac{n!m!}{\Delta}\right)  $. \ Then there exists a
\textbf{weakly-symmetric} proper verifier $V_{A}^{\ast}$\ (not necessarily
computationally efficient) such that, with probability $1-O\left(
\delta\right)  $\ over $A$, the following holds:%
\begin{align*}
\Pr_{S_{1},\ldots,S_{k}\sim\mathcal{D}_{A}}\left[  V_{A}^{\ast}\left(
S_{1},\ldots,S_{k}\right)  \text{ accepts}\right]   &  \geq1-\Delta,\\
\Pr_{S_{1},\ldots,S_{k}\sim\mathcal{U}}\left[  V_{A}^{\ast}\left(
S_{1},\ldots,S_{k}\right)  \text{ accepts}\right]   &  \leq\Delta.
\end{align*}

\end{theorem}

\begin{proof}
Let $V_{A}^{\ast}\left(  S_{1},\ldots,S_{k}\right)  $ accept if and only if
there \textit{exist} permutations $\sigma\in S_{m}$\ and $\tau\in S_{n}$\ that
cause $V_{P_{\sigma}AP_{\tau}}\left(  S_{1},\ldots,S_{k}\right)  $\ to accept,
where $V_{A}$\ is the verifier from Corollary \ref{ampcor}. \ Then clearly
$V_{A}^{\ast}$\ is weakly-symmetric: i.e., it behaves identically on $A$\ and
$P_{\sigma}AP_{\tau}$, for any pair $\left(  \sigma,\delta\right)  \in
S_{m}\times S_{n}$. \ Moreover, by Corollary \ref{ampcor}\ together with the
union bound, we have%
\begin{align*}
\Pr_{S_{1},\ldots,S_{k}\sim\mathcal{D}_{A}}\left[  V_{A}^{\ast}\left(
S_{1},\ldots,S_{k}\right)  \text{ rejects}\right]   &  \leq n!m!\left(
\frac{1}{c^{k}}\right)  \leq\Delta,\\
\Pr_{S_{1},\ldots,S_{k}\sim\mathcal{U}}\left[  V_{A}^{\ast}\left(
S_{1},\ldots,S_{k}\right)  \text{ accepts}\right]   &  \leq n!m!\left(
\frac{1}{c^{k}}\right)  \leq\Delta
\end{align*}
with probability $1-O\left(  \delta\right)  $\ over $A$.
\end{proof}

\section{Detecting Deviations from Uniformity\label{DETECT}}

In Section \ref{DEV}, we showed that most \textsc{BosonSampling}%
\ distributions $\mathcal{D}_{A}$\ have constant\ variation distance from the
uniform distribution $\mathcal{U}$\ over experimental outcomes $S$\ (and
furthermore, that this is true even if we restrict to collision-free outcomes
$S\in\Lambda_{m,n}$). \ However, we did not show how to distinguish
$\mathcal{D}_{A}$\ from $\mathcal{U}$\ using a polynomial-time algorithm.
\ Moreover, the distinguishing procedure $V_{A}$\ from Corollary \ref{ampcor}
involved computing\ $\left\vert \operatorname*{Per}\left(  A_{S}\right)
\right\vert ^{2}$ for each experimental outcome $S$. \ And just as Gogolin et
al. \cite{gogolin}\ asserted, we do \textit{not} expect $\left\vert
\operatorname*{Per}\left(  A_{S}\right)  \right\vert ^{2}$ to computable (or
even closely approximable) in polynomial time---at least, to whatever extent
we expect \textsc{BosonSampling}\ to be a hard problem in the first
place.\footnote{Indeed, we even conjecture that this problem is
\textit{nonuniformly} hard: that is, for a Haar-random $A\in\mathbb{C}%
^{m\times n}$, we conjecture that there is no polynomial-size circuit $C_{A}%
$\ that computes or closely approximates $\left\vert \operatorname*{Per}%
\left(  A_{S}\right)  \right\vert ^{2}$ given $S\in\Lambda_{m,n}$ as input.}

Nevertheless, in this section we will show that $\mathcal{D}_{A}%
$\ \textit{can} be distinguished in polynomial time from $\mathcal{U}$---and
moreover, quite easily. \ As we will discuss in Section \ref{MOCKUP}, one
interpretation of this result is that $\mathcal{D}_{A}$\ versus $\mathcal{U}$
was simply the \textquotedblleft wrong comparison\textquotedblright: instead,
one should compare $\mathcal{D}_{A}$\ to various more interesting
\textquotedblleft mockup\ distributions,\textquotedblright\ which are easy to
sample classically but at least encode \textit{some} information about $A$.
\ On the other hand, our result suffices to refute the claim of Gogolin et
al.\ \cite{gogolin},\ about $\mathcal{D}_{A}$\ looking just like the uniform
distribution\ to any polynomial-time algorithm.

Given a matrix $X\in\mathbb{C}^{n\times n}$, let $R_{i}\left(  X\right)  $\ be
the squared $2$-norm of $X$'s $i^{th}$\ row:%
\[
R_{i}=R_{i}\left(  X\right)  =\left\vert x_{i1}\right\vert ^{2}+\cdots
+\left\vert x_{in}\right\vert ^{2}.
\]
Clearly $\operatorname{E}\left[  R_{i}\right]  =\operatorname*{Var}\left[
R_{i}\right]  =n$\ if $X\sim\mathcal{N}$ is Gaussian. \ Indeed, it is not hard
to see that $R_{i}$\ is a $\chi^{2}$ random variable with $2n$ degrees of
freedom, multiplied by the scalar $1/2$.\ \ Its pdf is%
\[
f_{R_{i}}\left(  r\right)  =\frac{e^{-r}r^{n-1}}{\left(  n-1\right)  !}.
\]
Throughout this paper, we will refer to such a random variable as a
\textit{complex }$\chi^{2}$\textit{ variable with }$n$\textit{ degrees of
freedom}.

Next, let $R$\ be the product of the squared row-norms:%
\[
R=R\left(  X\right)  =R_{1}\cdots R_{n}.
\]
Note that if $X\sim\mathcal{N}$ is Gaussian, then by the independence of the
rows,%
\[
\operatorname{E}\left[  R\right]  =\operatorname{E}\left[  R_{1}\right]
\cdots\operatorname{E}\left[  R_{n}\right]  =n^{n}.
\]
Thus, it will be convenient to define%
\[
R^{\ast}=\frac{R}{n^{n}}%
\]
so that $\operatorname{E}\left[  R^{\ast}\right]  =1$. \ Of course, $R^{\ast}$
can be computed in linear (i.e., $O\left(  n^{2}\right)  $)\ time given
$X\in\mathbb{C}^{n\times n}$ as input.

Our claim is that \textit{computing }$R^{\ast}\left(  A_{S}\right)  $\textit{,
for a few experimental outcomes }$S\in\Lambda_{m,n}$\textit{, already suffices
to tell whether }$S$\textit{ was drawn from }$\mathcal{D}_{A}$\textit{\ or
from }$\mathcal{U}$\textit{, with high probability}. \ The intuition for this
claim is simple: first, each $R_{i}$\ is at least \textit{slightly} correlated
with $\left\vert \operatorname*{Per}\left(  X\right)  \right\vert ^{2}$, since
multiplying the $i^{th}$\ row of $X$ by any scalar $c$ also multiplies
$\operatorname*{Per}\left(  X\right)  $\ by $c$. \ Second, if $X$\ is
Gaussian, then $R_{i}\left(  X\right)  /n$\ will typically be $1\pm
\Theta\left(  1/\sqrt{n}\right)  $. \ This suggests that the \textit{product}
of the $R_{i}\left(  X\right)  /n$'s, over all $i\in\left[  n\right]  $, will
typically be $1\pm\Theta\left(  1\right)  $: or in other words, that $R^{\ast
}$\ will have constant-sized fluctuations. \ If so, then $R^{\ast}\left(
X\right)  $\ should be an easy-to-compute random variable, whose fluctuations
nontrivially correlate with the fluctuations in $\left\vert
\operatorname*{Per}\left(  X\right)  \right\vert ^{2}$. \ Therefore $R^{\ast
}\left(  A_{S}\right)  $\ should be larger, in expectation, if $S$ was drawn
from $\mathcal{D}_{A}$\ than if $S$ was drawn from $\mathcal{U}$. \ (See
Figure \ref{rstar} for plots of expected distribution over $R^{\ast}$, both
for a Haar-random $\mathcal{D}_{A}$\ and for $\mathcal{U}$.)

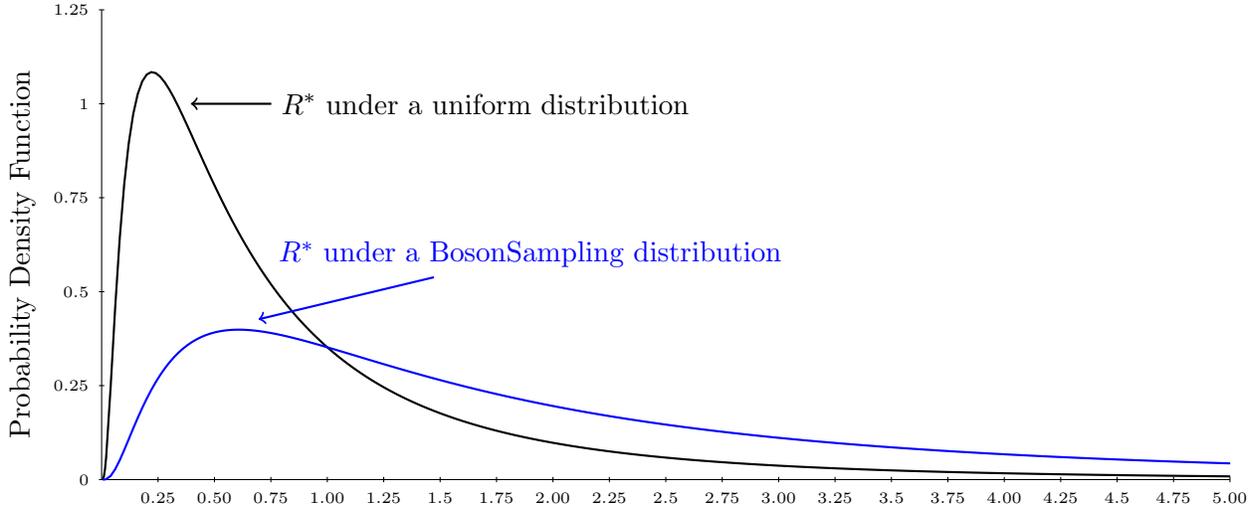
\begin{figure}[ptb]
\centering
\begin{tikzpicture}[x=3cm, y=5cm]
\draw (0,0) -- (5,0) node[right] {};
\draw (0,0) -- (0,1.25) node[above] {};
\foreach \x in
{0.25,0.50,0.75,1.00,1.25,1.5,1.75,2.00,2.25,2.5,2.75,3.00,3.25,3.5,3.75,4.00,4.25,4.5,4.75,5.00}
\draw (\x, 1pt) -- (\x,-1pt) node[anchor=north] {\tiny $\x$};
\foreach \y in {0, 0.25, 0.5, 0.75, 1, 1.25}
\draw (1pt,\y) -- (-1pt, \y) node[anchor=east] {\tiny $\y$};
\node[rotate=90] at (-0.35,0.6) {Probability Density Function};
\draw[color=Black, thick, mark=square*, mark size=0pt] plot
file {R_star_uniform.table};
		\draw[color=DarkBlue, thick, mark=square*, mark size=0pt] plot
		file {R_star_boson.table};
\node[text=Black] (Unif) at (1.7,1) {
$R^{\ast}$ under a uniform distribution};
\node (UnifPoint) at (0.35,1.0) {};
\draw[->, thick, color=Black] (Unif) -- (UnifPoint);
\node[text=DarkBlue] (Boson) at (1.9,0.6) {
$R^{\ast}$ under a BosonSampling distribution};
\node (BosonPoint) at (0.65,0.42) {};
\draw[->, thick, color=DarkBlue] (Boson) -- (BosonPoint);
\end{tikzpicture}
\caption{Probability density functions for the row-norm estimator $R^{\ast
}\left(  A_{S}\right)  $, when $S$ is drawn either from the uniform
distribution $\mathcal{U}$\ (in which case $R^{\ast}\left(  A_{S}\right)  $
quickly converges to a lognormal random variable), or from a Haar-random
\textsc{BosonSampling} distribution $\mathcal{D}_{A}$, in the limits
$n\rightarrow\infty$\ and $m/n\rightarrow\infty$. \ One can see that $R^{\ast
}$\ is typically larger under $\mathcal{D}_{A}$, and that $R^{\ast}$ suffices
to distinguish $\mathcal{D}_{A}$\ from $\mathcal{U}$\ with constant bias.}%
\label{rstar}%
\end{figure}

Before formalizing these intuitions, we make two remarks. \ First, $R^{\ast
}\left(  X\right)  $\ does not yield a \textit{good} approximation to
$\left\vert \operatorname*{Per}\left(  X\right)  \right\vert ^{2}$---nor
should we expect it to, if we believe (following \cite{aark})\ that
approximating $\left\vert \operatorname*{Per}\left(  X\right)  \right\vert
^{2}$, with $1-1/\operatorname*{poly}\left(  n\right)  $\ probability and to
within $\pm n!/\operatorname*{poly}\left(  n\right)  $\ error, is a
$\mathsf{\#P}$-hard problem! \ Instead, $R^{\ast}\left(  X\right)  $\ is
\textquotedblleft just barely\textquotedblright\ correlated with $\left\vert
\operatorname*{Per}\left(  X\right)  \right\vert ^{2}$. \ Still, we will show
that this correlation is already enough to distinguish $\mathcal{D}_{A}$\ from
$\mathcal{U}$\ with constant bias. \ Second, if we were satisfied to
distinguish $\mathcal{D}_{A}$\ from $\mathcal{U}$\ with
$1/\operatorname*{poly}\left(  n\right)  $ bias, then it would even suffice to
look at $R_{i}\left(  X\right)  $, for any single row $i$. \ Indeed, we
conjecture that it would even suffice to look at $\left\vert x_{ij}\right\vert
^{2}$, for any single \textit{entry} $x_{ij}$\ of $X$! \ However, looking at
the product of the $R_{i}$'s seems necessary if we want a constant-bias distinguisher.

We now prove that looking at $R^{\ast}\left(  A_{S}\right)  $\ indeed
distinguishes $\mathcal{D}_{A}$\ from $\mathcal{U}$, with high probability
over $A\in\mathbb{C}^{m\times n}$. \ As in Section \ref{DEV}, let%
\[
P=P\left(  X\right)  =\frac{\left\vert \operatorname*{Per}\left(  X\right)
\right\vert ^{2}}{n!}.
\]
Then let%
\[
Q=Q\left(  X\right)  =\frac{P\left(  X\right)  }{R^{\ast}\left(  X\right)  }.
\]
The following proposition, though trivial to prove, is crucial to what follows.

\begin{proposition}
\label{indprop}$Q\left(  X\right)  $ and $R^{\ast}\left(  X\right)  $ are
independent random variables, if $X\sim\mathcal{N}\left(  0,1\right)
_{\mathbb{C}}^{n\times n}$ is Gaussian. \ Indeed, we can think of $Q$\ as
simply $\left\vert \operatorname*{Per}\left(  Y\right)  \right\vert ^{2}/n!$,
where $Y\in\mathbb{C}^{n\times n}$ is a matrix each of whose rows is an
independent Haar-random vector, with squared $2$-norm equal to $n$. \ So in
particular, $\operatorname{E}\left[  Q\right]  =1$.
\end{proposition}

\begin{proof}
Suppose we first choose $Y$ as above, and then choose $R_{1},\ldots,R_{n}$\ as
independent complex $\chi^{2}$ random variables with $n$ degrees of freedom.
\ Then $X$, a Gaussian random matrix, can be obtained from $Y$ by simply
multiplying the $i^{th}$ row of $Y$ by $\sqrt{R_{i}/n}$, for all $i\in\left[
n\right]  $. \ So we can decompose $P$ as%
\[
P=\frac{\left\vert \operatorname*{Per}\left(  X\right)  \right\vert ^{2}}%
{n!}=\frac{\left\vert \operatorname*{Per}\left(  Y\right)  \right\vert ^{2}%
}{n!}\frac{R_{1}\cdots R_{n}}{n^{n}}=QR^{\ast}%
\]
where $Q=\left\vert \operatorname*{Per}\left(  Y\right)  \right\vert ^{2}/n!$
and $R^{\ast}=R_{1}\cdots R_{n}/n^{n}$ are independent by construction.
\end{proof}

As in Section \ref{DEV}, let $\mathcal{N}=\mathcal{N}\left(  0,1\right)
_{\mathbb{C}}^{n\times n}$\ be the Gaussian distribution over $X\in
\mathbb{C}^{n\times n}$, let $\mathcal{H}$\ be the Gaussian distribution
scaled by $\left\vert \operatorname*{Per}\left(  X\right)  \right\vert ^{2}$,
and let $f_{\mathcal{N}}\left(  X\right)  $\ and $f_{\mathcal{H}}\left(
X\right)  =f_{\mathcal{N}}\left(  X\right)  P\left(  X\right)  $\ be the pdfs
of $\mathcal{N}$\ and $\mathcal{H}$\ respectively. \ Then the following lemma
says that, provided $\operatorname*{E}\nolimits_{\mathcal{N}}\left[
\left\vert R^{\ast}\left(  X\right)  -1\right\vert \right]  $\ is large, we
can indeed use $R^{\ast}\left(  X\right)  $\ to distinguish the case
$X\sim\mathcal{H}$\ from the case $X\sim\mathcal{N}$:

\begin{lemma}
\label{rlem}%
\[
\Pr_{\mathcal{H}}\left[  R^{\ast}\geq1\right]  -\Pr_{\mathcal{N}}\left[
R^{\ast}\geq1\right]  =\frac{1}{2}\operatorname*{E}\limits_{\mathcal{N}%
}\left[  \left\vert R^{\ast}-1\right\vert \right]  .
\]

\end{lemma}

\begin{proof}
We have%
\begin{align}
\Pr_{\mathcal{H}}\left[  R^{\ast}\left(  X\right)  \geq1\right]
-\Pr_{\mathcal{N}}\left[  R^{\ast}\left(  X\right)  \geq1\right]   &
=\int_{X:R^{\ast}\left(  X\right)  \geq1}\left(  f_{\mathcal{H}}\left(
X\right)  -f_{\mathcal{N}}\left(  X\right)  \right)  dX\nonumber\\
&  =\int_{X:R^{\ast}\left(  X\right)  \geq1}\left(  f_{\mathcal{N}}\left(
X\right)  P\left(  X\right)  -f_{\mathcal{N}}\left(  X\right)  \right)
dX\nonumber\\
&  =\int_{X:R^{\ast}\left(  X\right)  \geq1}f_{\mathcal{N}}\left(  X\right)
\left(  Q\left(  X\right)  R^{\ast}\left(  X\right)  -1\right)  dX\nonumber\\
&  =\int_{X:R^{\ast}\left(  X\right)  \geq1}f_{\mathcal{N}}\left(  X\right)
\left(  R^{\ast}\left(  X\right)  -1\right)  dX. \label{indline}%
\end{align}
Here line (\ref{indline}) used the fact that $Q$ and $R^{\ast}$ are
independent, and hence that%
\[
\operatorname*{E}\left[  Q|R^{\ast}=r\right]  =\operatorname*{E}\left[
Q\right]  =1
\]
regardless of the value of $R^{\ast}$.

Now, since $\operatorname*{E}\left[  R^{\ast}\right]  =1$, the above also
equals%
\[
\int_{X:R^{\ast}\left(  X\right)  <1}f_{\mathcal{N}}\left(  X\right)  \left(
1-R^{\ast}\left(  X\right)  \right)  dX.
\]
But this means that it's simply%
\[
\frac{1}{2}\int_{X}f_{\mathcal{N}}\left(  X\right)  \left\vert R^{\ast}\left(
X\right)  -1\right\vert dX=\frac{1}{2}\operatorname*{E}\limits_{\mathcal{N}%
}\left[  \left\vert R^{\ast}-1\right\vert \right]  .
\]

\end{proof}

Thus, we arrive the remarkable result that, to understand the ability of our
row-norm estimator $R^{\ast}$ to distinguish the distributions $\mathcal{H}%
$\ and $\mathcal{N}$, \textit{we need only understand the intrinsic variation
of }$R^{\ast}$. \ All other information about $\operatorname*{Per}\left(
X\right)  $\ is irrelevant. \ (Of course, this result depends crucially on the
independence of $Q$ and $R^{\ast}$.)

Better yet, $R=R_{1}\cdots R_{n}$ is a product of $n$\ iid random
variables---which means that, with a little work, we can indeed understand the
distribution of $R$. \ As a first step, let%
\[
L:=\ln R.
\]
Then $L$, of course, is a \textit{sum} of $n$ iid random variables,%
\[
L=\ln R_{1}+\cdots+\ln R_{n}.
\]
Indeed, let $\ell_{n}$ be the log of a complex $\chi^{2}$ variable with $n$
degrees of freedom: that is,
\[
\ell_{n}=\ln\left(  \left\vert x_{1}\right\vert ^{2}+\cdots+\left\vert
x_{n}\right\vert ^{2}\right)  ,
\]
where $x_{j}\sim\mathcal{N}\left(  0,1\right)  _{\mathbb{C}}$ for all
$j\in\left[  n\right]  $. \ Then each $\ln R_{i}$\ is distributed as $\ell
_{n}$, so $L$ is simply the sum of $n$\ independent $\ell_{n}$\ random
variables. \ As such, we might expect $L$ to be close to Gaussian. \ The
Berry-Esseen Theorem (which we state for completeness) assures us that this is
indeed the case, provided the $\ell_{n}$'s satisfy certain conditions.

\begin{theorem}
[Berry-Esseen Theorem]\label{berryesseen}Let $Z_{1},\ldots,Z_{n}$\ be real iid
random variables satisfying%
\begin{align*}
\operatorname{E}\left[  Z_{i}\right]   &  =\upsilon,\\
\operatorname{E}\left[  \left(  Z_{i}-\upsilon\right)  ^{2}\right]   &
=\sigma^{2}>0,\\
\operatorname{E}\left[  \left\vert Z_{i}-\upsilon\right\vert ^{3}\right]   &
=\rho<\infty.
\end{align*}
Then let%
\[
Z:=Z_{1}+\cdots+Z_{n},
\]
and let $W\sim\mathcal{N}\left(  \upsilon n,\sigma^{2}n\right)  _{\mathbb{R}}%
$\ be a real Gaussian with mean $\upsilon n$\ and variance $\sigma^{2}n$.
\ Then for all $x\in\mathbb{R}$,%
\[
\left\vert \Pr\left[  Z\leq x\right]  -\Pr\left[  W\leq x\right]  \right\vert
\leq\frac{C\rho}{\sigma^{3}\sqrt{n}},
\]
where $C$ is some universal constant.
\end{theorem}

Thus, it remains only to check that the random variables $\ell_{n}$\ satisfy
the conditions for Theorem \ref{berryesseen}. \ We can do so using the
following striking relations:

\begin{lemma}
\label{heartbreak}We have%
\begin{align*}
\operatorname{E}\left[  \ell_{n}\right]   &  =-\gamma+\sum_{j=1}^{n-1}\frac
{1}{j}=\ln n-O\left(  \frac{1}{n}\right)  ,\\
\operatorname*{Var}\left[  \ell_{n}\right]   &  =\frac{\pi^{2}}{6}-\sum
_{j=1}^{n-1}\frac{1}{j^{2}}=\frac{1+o\left(  1\right)  }{n},\\
\operatorname{E}\left[  \left(  \ell_{n}-\operatorname{E}\left[  \ell
_{n}\right]  \right)  ^{4}\right]   &  =6\left(  \frac{\pi^{4}}{90}-\sum
_{j=1}^{n-1}\frac{1}{j^{4}}\right)  +3\operatorname*{Var}\left[  \ell
_{n}\right]  ^{2}=\frac{3+o\left(  1\right)  }{n^{2}},
\end{align*}
where $\gamma=0.577\ldots$\ is the Euler-Mascheroni constant.
\end{lemma}

\begin{proof}
[Proof Sketch]The first step is to rewrite the moments in terms of the
so-called \textit{cumulants} $\kappa_{k}$\ of $\ell_{n}$, which are the
coefficients of the log of $\ell_{n}$'s moment generating function:%
\[
\sum_{k=1}^{\infty}\kappa_{k}\frac{t^{k}}{k!}=\ln\operatorname{E}\left[
e^{\ell_{n}t}\right]  =\ln\operatorname{E}\left[  \left(  \left\vert
x_{1}\right\vert ^{2}+\cdots+\left\vert x_{n}\right\vert ^{2}\right)
^{t}\right]  .
\]
By standard facts about cumulants, we have%
\begin{align*}
\operatorname{E}\left[  \ell_{n}\right]   &  =\kappa_{1},\\
\operatorname*{Var}\left[  \ell_{n}\right]   &  =\kappa_{2},\\
\operatorname{E}\left[  \left(  \ell_{n}-\operatorname{E}\left[  \ell
_{n}\right]  \right)  ^{4}\right]   &  =\kappa_{4}+3\kappa_{2}^{2}.
\end{align*}
The second step is to express the cumulants in terms of the so-called
\textit{polygamma function} $\phi^{\left(  k\right)  }$:%
\[
\phi^{\left(  k\right)  }\left(  z\right)  :=\left(  \frac{d}{dz}\right)
^{k+1}\left(  \ln\Gamma\left(  z\right)  \right)  .
\]
One can show that%
\[
\kappa_{k}=\phi^{\left(  k-1\right)  }\left(  n\right)  .
\]
The third step is to note that, for positive integers $n$ and $k\geq2$, we
have%
\[
\phi^{\left(  k-1\right)  }\left(  n\right)  =\left(  -1\right)  ^{k}\left(
k-1\right)  !\sum_{j=n}^{\infty}\frac{1}{j^{k}}=\left(  -1\right)  ^{k}\left(
k-1\right)  !\left(  \zeta\left(  k\right)  -\sum_{j=1}^{n-1}\frac{1}{j^{k}%
}\right)  ,
\]
where $\zeta\left(  k\right)  =\sum_{j=1}^{\infty}\frac{1}{j^{k}}$\ is the
Riemann zeta function (which satisfies $\zeta\left(  2\right)  =\pi^{2}/6$ and
$\zeta\left(  4\right)  =\pi^{4}/90$). \ On the other hand, noting that
$\zeta\left(  1\right)  $\ diverges, when $k=1$\ we instead have%
\[
\phi^{\left(  0\right)  }\left(  n\right)  =-\gamma+\sum_{j=1}^{n-1}\frac
{1}{j}.
\]

\end{proof}

Note that, for any real random variable $X$, by convexity we have%
\[
\operatorname{E}\left[  \left\vert X-\operatorname{E}\left[  X\right]
\right\vert ^{3}\right]  \leq\operatorname{E}\left[  \left(
X-\operatorname{E}\left[  X\right]  \right)  ^{4}\right]  ^{3/4}.
\]
So one corollary of Lemma \ref{heartbreak}\ is%
\[
\operatorname{E}\left[  \left\vert \ell_{n}-\operatorname{E}\left[  \ell
_{n}\right]  \right\vert ^{3}\right]  \leq\frac{3^{3/4}+o\left(  1\right)
}{n^{3/2}}.
\]
Thus, in the notation of Theorem \ref{berryesseen}, we have%
\[
\frac{\rho}{\sigma^{3}\sqrt{n}}=O\left(  \frac{1/n^{3/2}}{\left(  1/\sqrt
{n}\right)  ^{3}\sqrt{n}}\right)  =O\left(  \frac{1}{\sqrt{n}}\right)  .
\]
This yields the following:

\begin{theorem}
\label{closetogauss}Let $L=\ln R$, and let $L^{\prime}$\ be a real Gaussian
with mean $\operatorname{E}\left[  \ell_{n}\right]  \cdot n$\ and variance
$\operatorname*{Var}\left[  \ell_{n}\right]  \cdot n$. \ Then for all
$x\in\mathbb{R}$,%
\[
\left\vert \Pr\left[  L\leq x\right]  -\Pr\left[  L^{\prime}\leq x\right]
\right\vert =O\left(  \frac{1}{\sqrt{n}}\right)  .
\]

\end{theorem}

Theorem \ref{closetogauss}\ has the following corollary.

\begin{corollary}
\label{er1cor}We have%
\[
\Pr_{\mathcal{N}}\left[  \left\vert R^{\ast}-1\right\vert \geq\frac{1}%
{2}\right]  \geq0.586-O\left(  \frac{1}{\sqrt{n}}\right)  ,~~~~~\frac{1}%
{2}\operatorname*{E}\limits_{\mathcal{N}}\left[  \left\vert R^{\ast
}-1\right\vert \right]  \geq0.146-O\left(  \frac{1}{\sqrt{n}}\right)  .
\]

\end{corollary}

\begin{proof}
Letting $L^{\prime}$ be the real Gaussian from Theorem \ref{closetogauss},
observe that%
\[
\operatorname*{E}\left[  L^{\prime}\right]  =\operatorname{E}\left[  \ell
_{n}\right]  \cdot n=\left(  -\gamma+\sum_{j=1}^{n-1}\frac{1}{j}\right)  \cdot
n\in\left[  n\ln n-\gamma,n\ln n-\frac{1}{2}\right]
\]
and%
\[
\operatorname*{Var}\left[  L^{\prime}\right]  =\operatorname*{Var}\left[
\ell_{n}\right]  \cdot n=\left(  \frac{\pi^{2}}{6}-\sum_{j=1}^{n-1}\frac
{1}{j^{2}}\right)  \cdot n\geq1.
\]
So recalling that $R=n^{n}R^{\ast}$\ and $L=\ln R$, we have%
\begin{align*}
\Pr_{\mathcal{N}}\left[  \left\vert R^{\ast}-1\right\vert \geq\frac{1}%
{2}\right]   &  =\Pr_{\mathcal{N}}\left[  R\leq\frac{n^{n}}{2}\right]
+\Pr_{\mathcal{N}}\left[  R\geq\frac{3n^{n}}{2}\right] \\
&  =\Pr_{\mathcal{N}}\left[  L\leq n\ln n-\ln2\right]  +\Pr_{\mathcal{N}%
}\left[  L\geq n\ln n+\ln\frac{3}{2}\right] \\
&  \geq\Pr\left[  L^{\prime}\leq n\ln n-\ln2\right]  +\Pr\left[  L^{\prime
}\geq n\ln n+\ln\frac{3}{2}\right]  -O\left(  \frac{1}{\sqrt{n}}\right) \\
&  \geq\int_{-\infty}^{-\ln2+1/2}\frac{e^{-x^{2}/2}}{\sqrt{2\pi}}dx+\int%
_{\ln3/2+\gamma}^{\infty}\frac{e^{-x^{2}/2}}{\sqrt{2\pi}}dx-O\left(  \frac
{1}{\sqrt{n}}\right) \\
&  \geq0.586-O\left(  \frac{1}{\sqrt{n}}\right)
\end{align*}
and%
\[
\frac{1}{2}\operatorname*{E}\limits_{\mathcal{N}}\left[  \left\vert R^{\ast
}-1\right\vert \right]  \geq\frac{1}{4}\Pr_{\mathcal{N}}\left[  \left\vert
R^{\ast}-1\right\vert \geq\frac{1}{2}\right]  \geq0.146-O\left(  \frac
{1}{\sqrt{n}}\right)  .
\]

\end{proof}

In particular, combining Lemma \ref{rlem} with Corollary \ref{er1cor} gives%
\[
\Pr_{\mathcal{H}}\left[  R^{\ast}\geq1\right]  -\Pr_{\mathcal{N}}\left[
R^{\ast}\geq1\right]  \geq0.146-O\left(  \frac{1}{\sqrt{n}}\right)  ,
\]
meaning that $R^{\ast}\left(  X\right)  $\ does indeed distinguish the case
$X\sim\mathcal{H}$\ from the case $X\sim\mathcal{N}$\ with constant bias.

Having established that, the last step is to deduce that, if $A$ is a
Haar-random \textsc{BosonSampling}\ matrix with $m\gg n$, then with high
probability over $A$, the row-norm estimator $R^{\ast}\left(  X\right)
$\ \textit{also} distinguishes the case $S\sim\mathcal{D}_{A}$\ from the case
$S\sim\mathcal{U}$\ with constant bias. \ However, this step is precisely
parallel to the analogous step in Section \ref{DEV}: once again, we use the
fact that an $n\times n$\ submatrix of such an $A$\ is close in variation
distance to an iid Gaussian matrix. \ Again, the main technical complications
are that we want to restrict attention to collision-free $S$'s only, and that
we need to argue that different $n\times n$\ submatrices of $A$ are close to
independent, in order to show that quantities such as $\Pr_{S\sim
\mathcal{D}_{A}}\left[  R^{\ast}\left(  A_{S}\right)  \geq1\right]  $\ have
small variances. \ If these complications are handled in precisely the same
way as in Section \ref{DEV}, then we immediately obtain Theorem \ref{detect}
from Section \ref{RESULTS}: namely, that asking whether $R^{\ast}\left(
A_{S}\right)  \geq1$\ suffices to distinguish $S\sim\mathcal{D}_{A}$\ from
$S\sim\mathcal{U}$\ with constant bias, with probability $1-O\left(
\delta\right)  $\ over a Haar-random $A\in\mathbb{C}^{m\times n}$\ ($m\geq
n^{5.1}/\delta$) and for sufficiently large $n$.

\subsection{Classical Mockup Distribution\label{MOCKUP}}

We showed that the row-norm estimator $R^{\ast}$\ can distinguish a generic
\textsc{BosonSampling}\ distribution $\mathcal{D}_{A}$\ from the uniform
distribution $\mathcal{U}$---but not, of course, that $R^{\ast}$%
\ distinguishes $\mathcal{D}_{A}$\ from \textit{any} classically-samplable
distribution. \ And indeed, in this subsection we point out that there
\textit{are} natural distributions that are easy to sample classically, but
that $R^{\ast}$ fails to distinguish from $\mathcal{D}_{A}$.

Given a \textsc{BosonSampling}\ matrix $A\in\mathbb{C}^{m\times n}$, let
$A^{\#}$\ be the $m\times n$\ matrix whose $\left(  i,j\right)  $\ entry is
$\left\vert a_{ij}\right\vert ^{2}$. \ Also, given $S=\left(  s_{1}%
,\ldots,s_{m}\right)  \in\Phi_{m,n}$, let $A_{S}^{\#}$\ be the $n\times
n$\ matrix consisting of $s_{1}$\ copies of $A^{\#}$'s first row, $s_{2}%
$\ copies of $A^{\#}$'s second row, and so on (precisely analogous to $A_{S}%
$). \ Finally, let $\mathcal{M}_{A}$, the \textquotedblleft classical mockup
distribution for $A$,\textquotedblright\ be the distribution over $\Phi_{m,n}%
$\ defined by%
\[
\Pr_{\mathcal{M}_{A}}\left[  S\right]  =\frac{\operatorname*{Per}\left(
A_{S}^{\#}\right)  }{s_{1}!\cdots s_{m}!}.
\]
We claim that $\mathcal{M}_{A}$\ is easy to sample in classical polynomial
time (in fact, in $O\left(  mn\right)  $\ time). \ To do so, for each
$j:=1$\ to $n$, just sample an index $h_{j}\in\left[  m\right]  $\ from the
probability distribution%
\[
\left(  \left\vert a_{1j}\right\vert ^{2},\ldots,\left\vert a_{mj}\right\vert
^{2}\right)  .
\]
Then for all $i\in\left[  m\right]  $, let $s_{i}:=\left\vert \left\{
j:h_{j}=i\right\}  \right\vert $\ , and output\ $\left(  s_{1},\ldots
,s_{m}\right)  $ as $S$. \ It is not hard to see that this algorithm will
output a given $S=\left(  s_{1},\ldots,s_{m}\right)  $ with probability
exactly equal to $\operatorname*{Per}\left(  A_{S}^{\#}\right)  /s_{1}!\cdots
s_{m}!$\ (the numerator comes from summing over all $n!$ possible permutations
of $h_{j}$'s that yield $S$, while the denominator is the size of $S$'s
automorphism group). \ This also implies that $\mathcal{M}_{A}$\ was a
normalized probability distribution in the first place.

Note that there is a clear \textquotedblleft physical\textquotedblright%
\ reason why $\mathcal{M}_{A}$\ is samplable in classical polynomial time.
\ Namely, $\mathcal{M}_{A}$\ is just the distribution output by a
\textsc{BosonSampling}\ device with matrix $A$, if the $n$ input photons are
all \textit{distinguishable}---or equivalently, if they behave as classical
particles rather than as bosons. \ In other words, $\mathcal{M}_{A}$ simply
models $n$ balls being thrown independently into $m$ bins (possibly with a
different distribution for each ball).

On the other hand, we now observe that \textit{our row-norm estimator,
}$R^{\ast}\left(  A_{S}\right)  $, \textit{fails completely to distinguish\ }%
$\mathcal{D}_{A}$\textit{\ from its \textquotedblleft classical
mockup\textquotedblright\ }$\mathcal{M}_{A}$\textit{.} \ The reason is just
that $\Pr_{\mathcal{M}_{A}}\left[  S\right]  $\ is correlated with $R^{\ast
}\left(  A_{S}\right)  $\ in exactly the same way that%
\[
\Pr_{\mathcal{D}_{A}}\left[  S\right]  =\frac{\left\vert \operatorname*{Per}%
\left(  A_{S}\right)  \right\vert ^{2}}{s_{1}!\cdots s_{m}!}%
\]
is correlated with $R^{\ast}\left(  A_{S}\right)  $. \ Indeed, \textit{both}
of these probabilities can be written as the product of $R^{\ast}\left(
A_{S}\right)  $\ with a random variable (the permanent or absolute-squared
permanent of a row-normalized matrix, respectively) that is independent of
$R^{\ast}\left(  A_{S}\right)  $. \ As a result, any verification test only
involving $R^{\ast}$---for example, accepting $S$ if and only if $R^{\ast
}\left(  A_{S}\right)  \geq1$---will accept with exactly the same probability
for $S\sim\mathcal{M}_{A}$\ as for $S\sim\mathcal{D}_{A}$.

Note that there are other distributions, besides $\mathcal{M}_{A}$, with the
same two properties: that they can be sampled in classical polynomial time,
but that the row-norm estimator $R^{\ast}$\ fails completely to distinguish
them from $\mathcal{D}_{A}$. \ One nice example is the
\textsc{FermionSampling}\ distribution $\mathcal{F}_{A}$: see Appendix
\ref{FERMION} for a detailed definition, as well as an $O\left(
mn^{2}\right)  $\ classical sampling algorithm. \ In $\mathcal{F}_{A}$, the
probabilities of collision-free outcomes $S\in\Lambda_{m,n}$\ are given by
$\left\vert \operatorname*{Det}\left(  A_{S}\right)  \right\vert ^{2}$\ rather
than $\left\vert \operatorname*{Per}\left(  A_{S}\right)  \right\vert ^{2}$.
\ But since the determinant is affected by scaling of rows in exactly the same
way as the permanent,\ it follows that $\mathcal{F}_{A}$\ must satisfy the
same row-norm statistics as $\mathcal{D}_{A}$.

Yet another example---and perhaps the simplest---is the distribution
$\mathcal{B}_{A}$\ obtained by sampling $n$\ rows $h_{1},\ldots,h_{n}%
\in\left[  m\right]  $\ independently, with each row $h_{j}$ drawn from the
same distribution%
\[
\Pr\left[  h_{j}=h\right]  =\frac{\left\vert a_{h1}\right\vert ^{2}%
+\cdots+\left\vert a_{hn}\right\vert ^{2}}{n},
\]
and then outputting $s_{i}:=\left\vert \left\{  j:h_{j}=i\right\}  \right\vert
$ as $S=\left(  s_{1},\ldots,s_{m}\right)  $. \ Like $\mathcal{M}_{A}$, the
distribution $\mathcal{B}_{A}$\ is classically samplable in $O\left(
mn\right)  $\ time. \ But again, the probabilities in $\mathcal{B}_{A}$\ are
affected by row scaling in exactly the same way as the probabilities
in\ $\mathcal{D}_{A}$, $\mathcal{M}_{A}$, and $\mathcal{F}_{A}$.

What are we to make of the above observations? \ Arguably, they merely
underscore what we said from the beginning: that the row-norm estimator cannot
prove, by itself, that \textsc{BosonSampling} is being solved. \ Indeed, it
can't even be used to prove the presence of quantum interference in an alleged
\textsc{BosonSampling}\ device. \ If $R^{\ast}\left(  A_{S}\right)
$\ satisfies the expected statistics, then we know that the device's output is
\textit{not} uniform random noise---and moreover, that the device samples from
\textit{some} distribution that depends nontrivially on the actual entries of
$A$, in way consistent with correct \textsc{BosonSampling}. \ If this were
combined with other evidence---e.g., verification with smaller numbers of
photons, verification that the multi-photon collisions satisfy the expected
statistics, and direct ruling out of alternatives such as $\mathcal{M}_{A}%
$---it would arguably provide circumstantial evidence that the device was
working properly, even with hundreds or thousands of photons.

Along those lines, we now observe that, if we only want to verify that a
\textsc{BosonSampling}\ device is \textit{not} sampling from $\mathcal{M}_{A}%
$\ (or from any distribution close to $\mathcal{M}_{A}$\ in variation
distance), then almost certainly this can be done in classical polynomial
time. \ The reason is that the probabilities in $\mathcal{M}_{A}$\ are given
by permanents of $n\times n$\ \textit{nonnegative} matrices---but such
permanents can be approximated to within $\varepsilon$\ multiplicative
error\ in $\operatorname*{poly}\left(  n,1/\varepsilon\right)  $ time, using
the famous randomized algorithm of Jerrum, Sinclair, and Vigoda \cite{jsv}%
.\ \ Thus, given experimental outcomes $S_{1},\ldots,S_{k}\in\Lambda_{m,n}$,
in classical polynomial time we can approximate $\operatorname*{Per}\left(
A_{S_{1}}^{\#}\right)  ,\ldots,\operatorname*{Per}\left(  A_{S_{k}}%
^{\#}\right)  $, then check whether they satisfy the statistics that we expect
if the $S_{i}$'s were drawn from $\mathcal{M}_{A}$. \ For similar but even
simpler reasons, we can almost certainly rule out, in classical polynomial
time, that a \textsc{BosonSampling}\ device\ is sampling from the particular
\textquotedblleft mockup\textquotedblright\ distributions $\mathcal{F}_{A}%
$\ or $\mathcal{B}_{A}$.

Admittedly, to make the argument completely rigorous, we would need to prove
that, with high probability over a Haar-random $A$, the \textsc{BosonSampling}%
\ distribution\ $\mathcal{D}_{A}$\ does \textit{not} give rise to the same
statistics for $\operatorname*{Per}\left(  A_{S}^{\#}\right)  $\ as
$\mathcal{M}_{A}$ does, or the same statistics for $\left\vert
\operatorname*{Det}\left(  A_{S}\right)  \right\vert ^{2}$\ as $\mathcal{F}%
_{A}$ does, or the same statistics for $\Pr_{\mathcal{B}_{A}}\left[  S\right]
$\ as\ $\mathcal{B}_{A}$ does. \ These statements are presumably true, but we
leave their proofs to future work.

\section{Summary and Open Problems\label{CONC}}

We began this paper by considering certain claims about \textsc{BosonSampling}%
\ made by Gogolin et al.\ \cite{gogolin}. \ We found those claims to be
misleading on at least three different levels. \ First, when testing a
\textsc{BosonSampling}\ device's output against the theoretical predictions,
there is not the slightest reason to ignore the labels of the modes---and once
the mode labels are accounted for, Gogolin et al.'s entire argument for the
\textquotedblleft near-uniformity\textquotedblright\ of the output
distribution collapses. \ Second, the observation that certifying a
\textsc{BosonSampling}\ distribution might be classically hard is not new; we
made it in \cite{aark}. \ And third, a \textsc{BosonSampling}\ device can
easily be \textit{faster} than its fastest classical certification, without
classical certification being practically impossible: indeed there exists a
regime (around $n=30$\ photons) where that is precisely what one would expect.
\ Moreover, we did not need any nontrivial technical work to reach these
conclusions: had we been content to refute Gogolin et al., this paper could
have been extremely short.

Having said that, Gogolin et al.'s paper does have the virtue that it suggests
interesting questions about the statistical aspects of \textsc{BosonSampling}.
\ So in the remainder of the paper, we took the opportunity to address some of
those questions. \ First, we proved what had previously been known only
heuristically: that given a Haar-random $A\in\mathbb{C}^{m\times n}$, with
high probability the\ \textsc{BosonSampling}\ distribution $\mathcal{D}_{A}%
$\ will have noticeable \textquotedblleft fluctuations\textquotedblright\ in
the probabilities\ (some outcomes being more likely, others less), which
easily distinguish $\mathcal{D}_{A}$ from the uniform distribution
$\mathcal{U}$. \ More surprisingly, we showed that $\mathcal{D}_{A}$\ can even
be distinguished from $\mathcal{U}$ in \textit{classical polynomial time},
using a simple row-norm estimator.\ \ As we pointed out in Section
\ref{INTRACT}, \textsc{BosonSampling}\ experiments with (say) $n=30$\ photons
could still be feasible and interesting, even if distinguishing $\mathcal{D}%
_{A}$\ from $\mathcal{U}$\ were asymptotically hard---but, as it turns out, it
isn't hard.

Needless to say, many open problems remain.

\begin{enumerate}
\item[(1)] Given a Gaussian matrix $X\sim\mathcal{N}\left(  0,1\right)
_{\mathbb{C}}^{n\times n}$, we showed that, in polynomial time, one can
compute a quantity $R\left(  X\right)  $\ that slightly but nontrivially
correlates with $\left\vert \operatorname*{Per}\left(  X\right)  \right\vert
^{2}$. \ This raises the obvious question: how well \textit{can} one
approximate $\left\vert \operatorname*{Per}\left(  X\right)  \right\vert ^{2}%
$\ in polynomial time, for a large fraction of Gaussian $X$'s? \ Our
\textquotedblleft Permanent of Gaussians Conjecture\textquotedblright\ (PGC),
from \cite{aark}, said that approximating $\left\vert \operatorname*{Per}%
\left(  X\right)  \right\vert ^{2}$\ to within $\pm\varepsilon n!$\ additive
error,\ for a $1-\delta$\ fraction of $X$'s, should be $\mathsf{\#P}$-hard, if
we take $n+1/\varepsilon+1/\delta$\ as the \textquotedblleft input
length.\textquotedblright\ \ But there remains an enormous gap in parameters
between that hardness conjecture and the weak approximation algorithm given
here. \ So in particular, even if we assume the PGC, it's perfectly
conceivable that much better approximation algorithms for $\left\vert
\operatorname*{Per}\left(  X\right)  \right\vert ^{2}$\ exist.

As a simple example, Linial, Samorodnitsky, and Wigderson \cite{lsw} proposed
an approximation algorithm for the permanent that works by first normalizing
all of the rows to $1$, then normalizing all of the columns to $1$, then
normalizing the rows, and so on iteratively until convergence; then using the
final product of row- and column-multipliers to produce an estimate for
$\operatorname*{Per}\left(  X\right)  $.\footnote{Actually, Linial et al.'s
algorithm is for approximating $\operatorname*{Per}\left(  X\right)  $\ where
$X$ is a \textit{nonnegative} matrix, and it proceeds by iteratively
normalizing the $1$\textit{-}norms of the rows and columns (i.e., the sums of
the entries). \ However, one could easily adapt their algorithm to attempt to
approximate $\left\vert \operatorname*{Per}\left(  X\right)  \right\vert ^{2}%
$, where $X$ is an arbitrary matrix, by instead normalizing the $2$-norms of
the rows and columns.} \ How good of an approximation does their algorithm
produce, for a Gaussian matrix $X\sim\mathcal{N}\left(  0,1\right)
_{\mathbb{C}}^{n\times n}$? \ Can we show that, with high probability, it
produces a \textit{better} approximation than our algorithm from Section
\ref{DETECT}, which normalized the rows only? \ Note that the techniques of
Section \ref{DETECT}\ no longer work, since one \textit{cannot} decompose
$\operatorname*{Per}\left(  X\right)  $\ as a product of Linial et al.'s
estimator and an independent random variable.

\item[(2)] Can we prove that a generic \textsc{BosonSampling}\ distribution
$\mathcal{D}_{A}$\ can be distinguished, in classical polynomial time, from
the classical \textquotedblleft mockup\textquotedblright\ distribution
$\mathcal{M}_{A}$\ defined in Section \ref{MOCKUP}? \ In Section \ref{MOCKUP},
we sketched a polynomial-time estimator to distinguish $\mathcal{D}_{A}%
$\ from\ $\mathcal{M}_{A}$\ (based on the famous permanent approximation
algorithm of Jerrum, Sinclair, and Vigoda \cite{jsv}), but it remains to prove
that our estimator works---or even, for that matter, that $\left\Vert
\mathcal{D}_{A}-\mathcal{M}_{A}\right\Vert =\Omega\left(  1\right)  $\ with
high probability over $A$. \ One can ask similar questions about
distinguishing $\mathcal{D}_{A}$\ from $\mathcal{F}_{A}$\ and from
$\mathcal{B}_{A}$.

\item[(3)] Of course, the broader question is whether there exists an
efficiently-samplable mockup distribution, call it $\mathcal{M}_{A}^{\prime}$,
that cannot be distinguished from $\mathcal{D}_{A}$\ by \textit{any} classical
polynomial-time algorithm. \ In Appendix \ref{FIXED}, we report an observation
of Fernando Brandao (personal communication): that, for any \textit{fixed}
$k$, there exists an efficiently-samplable distribution that has large
min-entropy, but that cannot be distinguished from (the collision-free part
of) $\mathcal{D}_{A}$\ by any algorithm running in $O\left(  \left(
m+n\right)  ^{k}\right)  $\ time. \ However, while this sounds striking, it
says very little about \textsc{BosonSampling}\ specifically. \ Indeed, if we
removed the requirement of large min-entropy, then as we point out in Appendix
\ref{FIXED}, it would be trivial to construct such a \textquotedblleft
mockup\textquotedblright\ for \textit{any} distribution $\mathcal{D}$
whatsoever! \ So in our view, the \textquotedblleft real\textquotedblright%
\ question is whether or not there exists a \textit{single}
efficiently-samplable mockup distribution that works against \textit{all}
polynomial-time distinguishers.

\item[(4)] Everything in (1)-(3) can be asked anew, when we consider various
\textquotedblleft experimentally realistic\textquotedblright\ variations of
\textsc{BosonSampling}. \ For example, what happens if some random subset of
the photons gets lost between the sources and the detectors? \ Or what happens
if the inputs are Gaussian states, rather than single-photon Fock states? \ In
Appendix \ref{INITIAL}, we explain a simple reason why even in these cases,
the row-norm estimator $R$\ still easily distinguishes the
\textsc{BosonSampling}\ distribution from the uniform distribution
$\mathcal{U}$. \ On the other hand, it is possible that the distinguishability
between $\mathcal{D}_{A}$\ and $\mathcal{M}_{A}$,\ or between $\mathcal{D}%
_{A}$\ and various other \textquotedblleft mockup\textquotedblright%
\ distributions, decreases when these or other experimental errors are incorporated.

\item[(5)] For most of this paper, we restricted attention to collision-free
outcomes $S\in\Lambda_{m,n}$. \ There were several reasons for this: first, it
greatly simplified the calculations; second, it mirrored the hardness results
of \cite{aark}, which also used collision-free outcomes only; and third, any
positive results about, say, the distinguishability of $\mathcal{D}_{A}$\ and
$\mathcal{U}$\ only become \textit{stronger} under such a restriction.
\ However, what happens when we \textquotedblleft put multi-photon collisions
back in\textquotedblright? \ For example, can we give stronger distinguishing
and verification algorithms, by taking advantage of the collision outcomes
$S\in\Phi_{m,n}\setminus\Lambda_{m,n}$\ that we previously discarded?

\item[(6)] All of our results in this paper hold with high probability for a
Haar-random \textsc{BosonSampling}\ matrix $A\in\mathbb{C}^{m\times n}$, where
(say) $m\geq n^{5.1}$. \ What can we say if $A$ is arbitrary rather than
random? \ What about if $m$ is smaller---say, $O\left(  n^{2}\right)  $ or
even $O\left(  n\right)  $? \ Finally, by choosing $m\geq n^{5.1}p\left(
n\right)  $, all of our results can be made to hold with probability at least
$1-1/p\left(  n\right)  $\ over a Haar-random $A$,\ where $p$ is any desired
polynomial. \ Can we show that they hold even with probability $1-1/\exp
\left(  n\right)  $ over $A$, and even for fixed $m$?

\item[(7)] We showed that $\left\Vert \mathcal{D}_{A}-\mathcal{U}\right\Vert
=\Omega\left(  1\right)  $\ with high probability over $A$. \ Can we show the
stronger result that $\left\Vert \mathcal{D}_{A}-\mathcal{U}\right\Vert
=1-o\left(  1\right)  $?

\item[(8)] We show in Appendix \ref{FERMION}\ that, if $X\sim\mathcal{N}$\ is
Gaussian, then $\left\vert \operatorname*{Det}\left(  X\right)  \right\vert
^{2}$\ is $O\left(  1/\log^{3/2}n\right)  $-close to a lognormal random
variable. \ However, we know that this approximation must break down when
$\left\vert \operatorname*{Det}\left(  X\right)  \right\vert ^{2}$\ is either
very small or very large. \ So it would be nice to have a more detailed
understanding of the pdf for $\left\vert \operatorname*{Det}\left(  X\right)
\right\vert ^{2}$---among other things, because such an understanding might
provide clues about the pdf for $\left\vert \operatorname*{Per}\left(
X\right)  \right\vert ^{2}$,\ and hence about how to prove the Permanent
Anti-Concentration Conjecture (PACC) from \cite{aark}.
\end{enumerate}

\section{Acknowledgments}

We thank Shalev Ben-David and Charles Xu for helpful discussions, and
especially Fernando Brandao\ for allowing us to include the results in
Appendix \ref{FIXED}, and Salil Vadhan and Luca Trevisan for clarifying
aspects of \cite{ttv}. \ We also thank Oded Regev and John Watrous for code
used in generating Figures \ref{pccfig}\ and \ref{rstar}.

\bibliographystyle{plain}
\bibliography{thesis}

\section{Appendix: Fixed Polynomial-Size Distinguishers\label{FIXED}}

In this appendix, we report a result due to Fernando Brandao (with Brandao's
kind permission): namely, that one can apply a recent theorem of Trevisan et
al.\ \cite{ttv} to construct an efficiently-samplable distribution that

\begin{enumerate}
\item[(i)] has large entropy (in fact, large min-entropy), and

\item[(ii)] is indistinguishable from (the collision-free part of) a generic
\textsc{BosonSampling} distribution,\ by circuits of any \textit{fixed}
polynomial size.
\end{enumerate}

We first observe that, if condition (i) is dropped, then condition (ii) is
trivial to satisfy for \textit{any} distribution whatsoever: there is nothing
specific about \textsc{BosonSampling}\ here!\footnote{We thank Salil Vadhan
for this observation.} \ To see this: first, recall that two distributions
$\mathcal{D}$\ and $\mathcal{D}^{\prime}$\ are called $\varepsilon
$\textit{-indistinguishable} with respect to a circuit $C$ if%
\[
\left\vert \Pr_{x\sim\mathcal{D}}\left[  C\left(  x\right)  \text{
accepts}\right]  -\Pr_{x\sim\mathcal{D}^{\prime}}\left[  C\left(  x\right)
\text{ accepts}\right]  \right\vert \leq\varepsilon,
\]
and are $\varepsilon$-indistinguishable with respect to a \textit{class} of
circuits $\mathcal{C}$\ if they are $\varepsilon$-indistinguishable\ with
respect to every $C\in\mathcal{C}$. \ Now let $\mathcal{D}$\ be any
distribution over $\left\{  0,1\right\}  ^{n}$. \ Then choose $w$\ elements
independently with replacement from $\mathcal{D}$, and let $\mathcal{D}%
^{\prime}$\ be the uniform distribution over the resulting multiset (so in
particular, $H\left(  \mathcal{D}^{\prime}\right)  \leq\log_{2}n$).
\ Certainly $\mathcal{D}^{\prime}$\ can be sampled by a circuit of size
$O\left(  nw\right)  $: just hardwire the elements.

Now, by a Chernoff bound, for any \textit{fixed} circuit $C$, clearly
$\mathcal{D}$\ and $\mathcal{D}^{\prime}$\ are $\varepsilon$-indistinguishable
with respect to $C$, with probability at least $1-\exp\left(  -\varepsilon
^{2}w\right)  $\ over the choice of $\mathcal{D}^{\prime}$. \ But there are
\textquotedblleft only\textquotedblright\ $n^{O\left(  n^{k}\right)  }%
$\ Boolean circuits of size at most $n^{k}$. \ So by the union bound, by
simply choosing $w=\Omega\left(  \frac{n^{k}\log n}{\varepsilon^{2}}\right)
$, we can ensure that $\mathcal{D}$\ and $\mathcal{D}^{\prime}$\ are
$\varepsilon$-indistinguishable with respect\ to \textit{all} circuits of size
at most $n^{k}$, with high probability over $\mathcal{D}^{\prime}$.

Thus, we see that the one nontrivial aspect of Brandao's observation is that,
in the case of \textsc{BosonSampling}, the \textquotedblleft
mockup\textquotedblright\ distribution $\mathcal{D}^{\prime}$\ can itself have
large entropy. \ Even for this, however, we will need very little that is
specific to \textsc{BosonSampling}: only that a generic \textsc{BosonSampling}%
\ distribution $\mathcal{D}_{A}$\ is close to a distribution with large min-entropy.

Before proving Brandao's observation, we need some more definitions. \ Given a
probability distribution $\mathcal{D}=\left\{  p_{x}\right\}  _{x}$ over
$\left[  N\right]  $, recall that the \textit{min-entropy} of $\mathcal{D}%
$\ is%
\[
H_{\min}\left(  \mathcal{D}\right)  :=\min_{x\in\left[  N\right]  }\log
_{2}\frac{1}{p_{x}},
\]
while the $2$\textit{-entropy} is%
\[
H_{2}\left(  \mathcal{D}\right)  :=-\log_{2}\sum_{x\in\left[  N\right]  }%
p_{x}^{2}.
\]
One can show that $H_{\min}\left(  \mathcal{D}\right)  \leq H_{2}\left(
\mathcal{D}\right)  \leq H\left(  \mathcal{D}\right)  $ for all $\mathcal{D}$,
where $H\left(  \mathcal{D}\right)  $\ is the ordinary Shannon entropy. \ The
following useful lemma, which we prove for completeness, provides a partial
converse to the inequality $H_{\min}\left(  \mathcal{D}\right)  \leq
H_{2}\left(  \mathcal{D}\right)  $.

\begin{lemma}
[Renner and Wolf \cite{rennerwolf}]\label{min2}For any distribution
$\mathcal{D}$\ over $\left[  N\right]  $\ and $\varepsilon>0$, there exists a
distribution $\mathcal{D}^{\prime}$\ over $\left[  N\right]  $\ such that
$\left\Vert \mathcal{D}-\mathcal{D}^{\prime}\right\Vert \leq\varepsilon$\ and
$H_{\min}\left(  \mathcal{D}^{\prime}\right)  \geq H_{2}\left(  \mathcal{D}%
\right)  -\log_{2}\frac{1}{\varepsilon}$.
\end{lemma}

\begin{proof}
Let $p\geq1/N$\ be a \textquotedblleft cutoff\textquotedblright\ to be
determined later. \ We define $\mathcal{D}^{\prime}=\left\{  q_{x}\right\}
_{x}$\ by starting from $\mathcal{D}=\left\{  p_{x}\right\}  _{x}$, then
cutting all probabilities $p_{x}$ such that $p_{x}>p$ down to $p$, and finally
adding back the probability mass that we removed (i.e., increasing the $q_{x}%
$'s) in any way that maintains the property $q_{x}\leq p$\ for all
$x\in\left[  N\right]  $. \ The result clearly satisfies%
\[
H_{\min}\left(  \mathcal{D}^{\prime}\right)  \geq\log_{2}\frac{1}{p}%
\]
and also%
\[
\left\Vert \mathcal{D}-\mathcal{D}^{\prime}\right\Vert =\frac{1}{2}\sum
_{x\in\left[  N\right]  }\left\vert p_{x}-q_{x}\right\vert \leq\sum
_{x:p_{x}>p}\left(  p_{x}-p\right)  .
\]
Indeed, let us simply choose $p$\ such that%
\[
\sum_{x:p_{x}>p}\left(  p_{x}-p\right)  =\varepsilon.
\]
By continuity, it is clear that such a $p$ exists for every $\varepsilon>0$.

Now notice that%
\[
\frac{1}{2^{H_{2}\left(  \mathcal{D}\right)  }}=\sum_{x\in\left[  N\right]
}p_{x}^{2}\geq\sum_{x:p_{x}>p}p_{x}^{2}>p\sum_{x:p_{x}>p}p_{x}\geq
p\varepsilon\geq\frac{\varepsilon}{2^{H_{\min}\left(  \mathcal{D}^{\prime
}\right)  }},
\]
or rearranging,%
\[
H_{\min}\left(  \mathcal{D}^{\prime}\right)  \geq H_{2}\left(  \mathcal{D}%
\right)  -\log_{2}\frac{1}{\varepsilon}.
\]

\end{proof}

We now state the result of Trevisan et al.\ \cite{ttv} that we will use.

\begin{theorem}
[Trevisan-Tulsiani-Vadhan \cite{ttv}]\label{ttvthm}Let $\mathcal{D}$ be any
distribution over $\left\{  0,1\right\}  ^{n}$\ such that $H_{\min}\left(
\mathcal{D}\right)  \geq n-k$. \ Then for every $T$ and $\varepsilon>0$, there
exists a circuit of size $\left(  T+n\right)  \operatorname*{poly}\left(
2^{k},1/\varepsilon\right)  $\ that samples a distribution $\mathcal{D}%
^{\prime}$\ that has $H_{\min}\left(  \mathcal{D}^{\prime}\right)  \geq
n-k$\ and that is $\varepsilon$-indistinguishable from $\mathcal{D}$\ by
circuits of size at most $T$.
\end{theorem}

Let $\mathcal{D}_{A}^{\ast}$\ be the \textsc{BosonSampling}\ distribution
$\mathcal{D}_{A}$\ conditioned on a collision-free outcome (that is, on
$S\in\Lambda_{m,n}$). \ Then it remains to show that $\mathcal{D}_{A}^{\ast}%
$\ has large min-entropy, with high probability over $A$. \ To do this, we
first recall two results proved in \cite{aark}\ (Theorem 5.2\ and Lemma 8.8
respectively there).

\begin{theorem}
[\cite{aark}]\label{haarhide}Let $m\geq\frac{n^{5}}{\delta}\log^{2}\frac
{n}{\delta}$\ for any $\delta>0$, let $\mathcal{N}$\ and $\mathcal{N}^{\prime
}$\ be as defined in Section \ref{DEV}, and let $f_{\mathcal{N}}\left(
X\right)  $\ and $f_{\mathcal{N}^{\prime}}\left(  X\right)  $\ be the pdfs of
$\mathcal{N}$\ and $\mathcal{N}^{\prime}$\ respectively. \ Then for all
$X\in\mathbb{C}^{n\times n}$,%
\[
f_{\mathcal{N}^{\prime}}\left(  X\right)  \leq\left(  1+O\left(
\delta\right)  \right)  f_{\mathcal{N}}\left(  X\right)  .
\]

\end{theorem}

In particular, if (say) $m\geq n^{5.1}$ and $n$ is sufficiently large, then
Theorem \ref{haarhide}\ implies that $f_{\mathcal{N}^{\prime}}\left(
X\right)  \leq2f_{\mathcal{N}}\left(  X\right)  $.

\begin{lemma}
[\cite{aark}]\label{4thmoment}$\operatorname{E}_{X\sim\mathcal{N}}\left[
\left\vert \operatorname*{Per}\left(  X\right)  \right\vert ^{4}\right]
=\left(  n+1\right)  \left(  n!\right)  ^{2}$.
\end{lemma}

Combining Theorem \ref{haarhide}\ with Lemma \ref{4thmoment}, Brandao observed
the following.

\begin{lemma}
[Brandao]\label{brandaolem1}Let $m\geq n^{5.1}$, and let $A\in\mathbb{C}%
^{m\times n}$ be a Haar-random \textsc{BosonSampling}\ matrix. \ Then for
sufficiently large $n$ and for all $\delta>0$, with probability at least
$1-\delta$\ over $A$ we have%
\[
\sum_{S\in\Lambda_{m,n}}\left\vert \operatorname*{Per}\left(  A_{S}\right)
\right\vert ^{4}\leq\frac{2\left(  n+1\right)  !}{\delta m^{n}}.
\]

\end{lemma}

\begin{proof}
For any fixed $S\in\Lambda_{m,n}$, we have%
\begin{align*}
\operatorname*{E}\limits_{A}\left[  \left\vert \operatorname*{Per}\left(
A_{S}\right)  \right\vert ^{4}\right]   &  =\frac{1}{m^{2n}}\operatorname*{E}%
\limits_{X\sim\mathcal{N}^{\prime}}\left[  \left\vert \operatorname*{Per}%
\left(  X\right)  \right\vert ^{4}\right] \\
&  \leq\frac{2}{m^{2n}}\operatorname*{E}\limits_{X\sim\mathcal{N}}\left[
\left\vert \operatorname*{Per}\left(  X\right)  \right\vert ^{4}\right] \\
&  =\frac{2\left(  n+1\right)  \left(  n!\right)  ^{2}}{m^{2n}},
\end{align*}
where the second line uses Theorem \ref{haarhide}\ and the third uses Lemma
\ref{4thmoment}. \ Hence%
\[
\operatorname*{E}\limits_{A}\left[  \sum_{S\in\Lambda_{m,n}}\left\vert
\operatorname*{Per}\left(  A_{S}\right)  \right\vert ^{4}\right]  \leq
\binom{m}{n}\frac{2\left(  n+1\right)  \left(  n!\right)  ^{2}}{m^{2n}}%
\leq\frac{2\left(  n+1\right)  !}{m^{n}}.
\]
So by Markov's inequality,%
\[
\Pr_{A}\left[  \sum_{S\in\Lambda_{m,n}}\left\vert \operatorname*{Per}\left(
A_{S}\right)  \right\vert ^{4}>\frac{2\left(  n+1\right)  !}{\delta m^{n}%
}\right]  <\delta.
\]

\end{proof}

Combining Lemma \ref{brandaolem1} with Lemma \ref{min2}\ now yields the
following corollary.

\begin{corollary}
[Brandao]\label{hmincor}Let $m\geq n^{5.1}$, and let $A\in\mathbb{C}^{m\times
n}$ be a Haar-random \textsc{BosonSampling}\ matrix. \ Then for all
$\varepsilon,\delta>0$, with probability at least $1-\delta$\ over $A$, there
exists a distribution $\mathcal{D}^{\prime}$\ over $\Lambda_{m,n}$\ such that
$\left\Vert \mathcal{D}^{\prime}-\mathcal{D}_{A}^{\ast}\right\Vert
\leq\varepsilon$\ and%
\[
H_{\min}\left(  \mathcal{D}^{\prime}\right)  \geq\log_{2}\binom{m}{n}-\log
_{2}\frac{n}{\varepsilon\delta}-O\left(  1\right)  .
\]

\end{corollary}

\begin{proof}
By Lemma \ref{brandaolem1}, for sufficiently large $n$ and all $\delta>0$,
with probability at least $1-\delta$\ over $A$ we have%
\begin{align*}
H_{2}\left(  \mathcal{D}_{A}^{\ast}\right)   &  =-\log_{2}\sum_{S\in
\Lambda_{m,n}}\Pr_{\mathcal{D}_{A}^{\ast}}\left[  S\right]  ^{2}\\
&  \geq-\log_{2}\sum_{S\in\Lambda_{m,n}}\left(  2\Pr_{\mathcal{D}_{A}}\left[
S\right]  \right)  ^{2}\\
&  =-2-\log_{2}\sum_{S\in\Lambda_{m,n}}\left\vert \operatorname*{Per}\left(
A_{S}\right)  \right\vert ^{4}\\
&  \geq-2-\log_{2}\frac{2\left(  n+1\right)  !}{\delta m^{n}}\\
&  =\log_{2}\frac{m^{n}}{n!}-\log_{2}\left(  n+1\right)  -\log_{2}\frac
{8}{\delta}\\
&  \geq\log_{2}\binom{m}{n}-\log_{2}n-\log_{2}\frac{9}{\delta}.
\end{align*}
So suppose the above inequality holds. \ Then by Lemma \ref{min2}, for every
$\varepsilon>0$\ there exists a distribution $\mathcal{D}^{\prime}$\ over
$\Lambda_{m,n}$\ such that $\left\Vert \mathcal{D}^{\prime}-\mathcal{D}%
_{A}^{\ast}\right\Vert \leq\varepsilon$\ and%
\begin{align*}
H_{\min}\left(  \mathcal{D}^{\prime}\right)   &  \geq H_{2}\left(
\mathcal{D}_{A}^{\ast}\right)  -\log_{2}\frac{1}{\varepsilon}\\
&  \geq\log_{2}\binom{m}{n}-\log_{2}\frac{n}{\varepsilon\delta}-O\left(
1\right)  .
\end{align*}

\end{proof}

It remains only to combine Corollary \ref{hmincor} with Theorem \ref{ttvthm}.

\begin{theorem}
[Brandao]\label{brandaothm}Let $m\geq n^{5.1}$, and let $A\in\mathbb{C}%
^{m\times n}$ be a Haar-random \textsc{BosonSampling}\ matrix. \ Then with
probability at least $1-\delta$\ over $A$, for every $T$ and $\varepsilon>0$,
there exists a circuit of size $T\operatorname*{poly}\left(  n,1/\varepsilon
,1/\delta\right)  $\ that samples a distribution $\mathcal{D}^{\prime}$\ that
has%
\[
H_{\min}\left(  \mathcal{D}^{\prime}\right)  \geq\log_{2}\binom{m}{n}-\log
_{2}\frac{n}{\varepsilon\delta}-O\left(  1\right)  ,
\]
and that is $\varepsilon$-indistinguishable from $\mathcal{D}_{A}^{\ast}$\ by
circuits of size at most $T$.
\end{theorem}

\section{Appendix: Arbitrary Initial States\label{INITIAL}}

Suppose the input to a \textsc{BosonSampling}\ device does \textit{not}
consist of a single-photon Fock state in each mode,\ but of some
\textquotedblleft messier\textquotedblright\ pure or mixed state. \ For
example, perhaps some subset of photons are randomly lost to the environment,
or perhaps the input involves coherent states or other Gaussian states. \ One
might wonder: \textit{then} does the device's output distribution
\textquotedblleft flatten out\textquotedblright\ and become nearly uniform, as
Gogolin et al.\ \cite{gogolin}\ claimed?

In this appendix, we briefly explain why the answer is still no. \ We will
argue that, for essentially \textit{any} initial state, the row-norm estimator
$R^{\ast}$\ from Section \ref{DETECT}\ still \textquotedblleft
works,\textquotedblright\ in the sense that it still distinguishes the output
distribution from the uniform distribution with non-negligible (and typically
constant) bias. \ Our discussion will be at a \textquotedblleft physics level
of rigor.\textquotedblright

Let $k$ be the total number of photons. Then an arbitrary initial pure state
might involve a superposition of different $k$'s, like so:%
\begin{equation}
\left\vert \psi\right\rangle =\sum_{k=0}^{\infty}\sum_{S\in\Phi_{n,k}}%
\alpha_{S}\left\vert S\right\rangle . \label{phi}%
\end{equation}
(Later we will generalize to mixed states.) \ We make the following assumptions:

\begin{enumerate}
\item[(a)] $m\gg n^{2}$: i.e., there are many more output modes than input
modes, as assumed throughout this paper. \ Alternatively, if we identify input
and output modes, then we assume that all photons are initially concentrated
in the first $n$ of $m$ modes.

\item[(b)] With non-negligible probability, $m\gg k^{2}$ (i.e., there are many
more output modes than photons).

\item[(c)] With non-negligible probability, $k>0$\ (i.e., at least one photon
gets detected).
\end{enumerate}

Assumption (c) is obviously necessary, but we believe that a more
sophisticated analysis would let us weaken assumption (a) and remove
assumption (b).\footnote{On the other hand, we observe that assumption (a)
cannot be removed entirely. \ For suppose the initial state was the $k$-photon
maximally mixed state $I_{m,k}$ (for any $k$), which assigns equal probability
to each element of $\Phi_{m,k}$. \ Then one can show that, \textit{regardless}
of what unitary transformation $U$\ was applied, the final state would again
be $I_{m,k}$. \ And therefore, neither $R^{\ast}$\ nor any other estimator
could distinguish the output distribution from uniform.}

If we want to model photon losses in this framework, we can do so by simply
letting our initial state be a mixture of different $\left\vert \psi
\right\rangle $'s, corresponding to different combinations of photons that are
\textquotedblleft destined to get lost.\textquotedblright\ \ We do not attempt
to model more complicated loss processes (e.g., processes that depend on the
unitary), or other sources of error such as mode mismatch. \ Extending the
statistical analysis of \textsc{BosonSampling}\ to those cases\ is an
interesting challenge for future work.

In what follows, we assume some familiarity with \textit{Fock polynomials},
i.e., polynomials in photon creation operators (see \cite[Section 3.2]{aark}
for a computer-science-friendly introduction). \ In particular, the Fock
polynomial associated to the state $\left\vert \psi\right\rangle $ from
(\ref{phi}) is
\[
p_{\psi}\left(  x\right)  =\sum_{k=0}^{\infty}\sum_{S\in\Phi_{n,k}}%
\frac{\alpha_{S}x^{S}}{\sqrt{S!}},
\]
where $x=\left(  x_{1},\ldots,x_{m}\right)  $ and $S=\left(  s_{1}%
,\ldots,s_{m}\right)  $ (in this particular case, $s_{i}=0$\ if $i>n$).
\ Here\ we use the shorthands $x^{S}=x_{1}^{s_{1}}\cdots x_{m}^{s_{m}}$ and
$S!=s_{1}!\cdots s_{m}!$. \ Applying a unitary transformation $U$ results in
the rotated state $p_{\psi}\left(  Ux\right)  $. \ So the probability of
measuring some particular $k$-photon\ outcome $S\in\Phi_{m,k}$ is then%
\[
\Pr\left[  S\right]  =\frac{1}{S!}\left\vert \left\langle x^{S},p_{\psi
}\left(  Ux\right)  \right\rangle \right\vert ^{2},
\]
where $\left\langle ,\right\rangle $\ represents the Fock inner product. \ We
can now use the adjoint property of the Fock product (Theorem 3.4 in
\cite{aark}) to move $U$ to the left-hand side:%
\begin{align*}
\Pr\left[  S\right]   &  =\frac{1}{S!}\left\vert \left\langle \left(  U^{\dag
}x\right)  ^{S},p_{\psi}\left(  x\right)  \right\rangle \right\vert ^{2}\\
&  =\frac{1}{S!}\left\vert \left\langle
{\textstyle\prod\limits_{i=1}^{m}}
\left(  U^{\dag}x_{i}\right)  ^{s_{i}},p_{\psi}\left(  x\right)  \right\rangle
\right\vert ^{2}.
\end{align*}
Because $\left\vert \psi\right\rangle $\ has support only on the first $n$
modes, the right-hand side uses only the variables $x_{1},\ldots,x_{n}$, so
only the first $n\ll\sqrt{m}$ rows of $U$ affect $\Pr\left[  S\right]  $.
\ Moreover, since only $i$'s such that $s_{i}>0$\ contribute to the contribute
to the left-hand side, only those columns of $U$ matter, of which there are at
most $k\ll\sqrt{m}$ (since $S$ is a $k$-photon outcome).

So by Theorem \ref{trunc}, we can achieve a good approximation to $\Pr\left[
S\right]  $\ by replacing $U^{\dag}$\ with an $n\times n$\ scaled Gaussian
matrix, $G/\sqrt{m}$\ where $G\sim\mathcal{N}\left(  0,1\right)  _{\mathbb{C}%
}^{n\times n}$. \ In that case,%
\[
\Pr\left[  S\right]  \approx\frac{1}{S!m^{k}}\left\vert \left\langle
{\textstyle\prod\limits_{i=1}^{m}}
\left(  Gx_{i}\right)  ^{s_{i}},p_{\psi}\left(  x\right)  \right\rangle
\right\vert ^{2}.
\]
We can then factor out the squared $2$-norms of the rows of $G$ (call them
$R_{1},\ldots,R_{n}$), resulting in a random $n\times n$\ matrix $G^{\prime}%
$\ remaining, each of whose rows is an independent, Haar-random unit vector:%
\begin{equation}
\Pr\left[  S\right]  \approx\left(
{\textstyle\prod\limits_{i=1}^{n}}
R_{i}^{s_{i}}\right)  \cdot\frac{1}{S!m^{k}}\left\vert \left\langle
{\textstyle\prod\limits_{i=1}^{m}}
\left(  G^{\prime}x_{i}\right)  ^{s_{i}},p_{\psi}\left(  x\right)
\right\rangle \right\vert ^{2}. \label{indep}%
\end{equation}
We have thus approximated $\Pr\left[  S\right]  $\ by the product of two
\textit{independent} random variables, the first of which is precisely the
row-norm estimator,%
\[
R_{S}:=%
{\textstyle\prod\limits_{i=1}^{n}}
R_{i}^{s_{i}}.
\]
Let%
\[
R_{S}^{\ast}=%
{\textstyle\prod\limits_{i=1}^{n}}
\frac{R_{i}^{s_{i}}}{n\left(  n+1\right)  \cdots\left(  n+s_{i}-1\right)  }%
\]
be the scaled version of $R_{S}$, so that $\operatorname*{E}\left[
R_{S}^{\ast}\right]  =1$. \ Also, let $\mathcal{U}$\ be the uniform
distribution over $S\in\Phi_{m,k}$, and let $\mathcal{D}$\ be the distribution
over $S\in\Phi_{m,k}$\ induced by \textsc{BosonSampling}\ with unitary
transformation $U$, and then conditioning on a $k$-photon outcome. \ Then just
like in Lemma \ref{rlem}\ of Section \ref{DETECT}, an immediate consequence of
the decomposition (\ref{indep}) is that%
\[
\Pr_{S\sim\mathcal{D}}\left[  R_{S}^{\ast}\geq1\right]  -\Pr_{S\sim
\mathcal{U}}\left[  R_{S}^{\ast}\geq1\right]  \approx\frac{1}{2}%
\operatorname*{E}\limits_{S\sim\mathcal{U}}\left[  \left\vert R_{S}^{\ast
}-1\right\vert \right]  .
\]
Or in other words, $R_{S}^{\ast}$\ will distinguish the case $S\sim
\mathcal{D}$\ from the case $S\sim\mathcal{U}$ with non-negligible bias,
\textit{if and only if} there is non-negligible \textquotedblleft intrinsic
variation\textquotedblright\ in $R_{S}^{\ast}$\ itself when $S\sim\mathcal{U}$.

Now, at least when $S$\ is collision-free (i.e., $s_{i}\in\left\{
0,1\right\}  $\ for all $i\in\left[  m\right]  $), our analysis in Section
\ref{DETECT}\ implies that $R_{S}^{\ast}$\ converges to a lognormal random
variable with%
\[
\operatorname*{E}\limits_{S\sim\mathcal{U}}\left[  \left\vert R_{S}^{\ast
}-1\right\vert \right]  =\Omega\left(  \frac{k}{n}\right)  .
\]
This is non-negligible whenever $k>0$,\ and constant whenever $k=\Omega\left(
n\right)  $. \ Moreover, one can show that if $S$ contains collisions, then
for a fixed $k$, the variation $\operatorname*{E}\left[  \left\vert
R_{S}^{\ast}-1\right\vert \right]  $\ only becomes \textit{larger} (we omit
the proof).

We can extend the above analysis to a mixed initial state $\rho$ by observing
that, in the mixed case, $\Pr\left[  S\right]  $\ is just a convex combination
of the probabilities arising from the pure states $\left\vert \psi
\right\rangle $ comprising $\rho$. \ So an identical term $R_{S}^{\ast}$\ can
be factored out from all of those probabilities.

\section{Appendix: FermionSampling\label{FERMION}}

In this appendix, we prove two results about \textsc{FermionSampling}: the
\textquotedblleft easier cousin\textquotedblright\ of \textsc{BosonSampling},
involving determinants rather than permanents. \ The first result is that, in
sharp contrast to what we conjecture for the bosonic case,
\textsc{FermionSampling} can be solved in classical polynomial time (indeed,
$O\left(  mn^{2}\right)  $ time), and by a particularly simple algorithm. \ It
was already known, from work by Valiant \cite{valiant:qc}, Terhal and
DiVincenzo \cite{td:fermion}, and Knill \cite{knill:matchgate}, that
\textsc{FermionSampling}\ is efficiently solvable classically. \ However, our
algorithm---which is basically identical to an algorithm discovered
independently in the field of \textit{determinantal point processes} (see Tao
\cite{tao:dpp}\ for example)---is both simpler and faster than the previous
\textsc{FermionSampling} algorithms.

The second result is that, if $X\sim\mathcal{N}\left(  0,1\right)
_{\mathbb{C}}^{n\times n}$\ is Gaussian, then $\left\vert \operatorname*{Det}%
\left(  X\right)  \right\vert ^{2}$\ converges at a $1/\log^{3/2}n$ rate to a
lognormal random variable. \ Again, this essentially follows from earlier
results by Girko \cite{girko} and Costello and Vu \cite{cv},\ except that they
considered real matrices only and did not bound the convergence rate. \ An
implication is that, as $n\rightarrow\infty$, the histogram of outcome
probabilities $\left\{  \Pr\left[  S\right]  \right\}  _{S\in\Lambda_{m,n}}$
for a Haar-random \textsc{FermionSampling}\ distribution converges to lognormal.

What does any of this have to do with the rest of the paper? \ There are three
connections. \ First, as discussed in Section \ref{MOCKUP},
\textsc{FermionSampling}\ is extremely interesting as a \textquotedblleft
mockup\textquotedblright\ of \textsc{BosonSampling}: something that satisfies
the same row-norm statistics (and, indeed, also involves quantum interference
among $n$ identical particles), but is nevertheless easy for a classical
computer. \ So it might be illuminating to see explicitly just \textit{why}
\textsc{FermionSampling}\ is so easy. \ Second, we will see how the
mathematical techniques developed in Sections \ref{DEV}\ and \ref{DETECT}\ for
\textsc{BosonSampling}, can easily be reapplied to \textsc{FermionSampling}.
\ And third, recall that the pdfs for $\left\vert \operatorname*{Det}\left(
X\right)  \right\vert ^{2}$\ and $\left\vert \operatorname*{Per}\left(
X\right)  \right\vert ^{2}$ look almost identical when plotted (see Figure
\ref{pccfig}). \ For this reason, \textsc{FermionSampling}\ can serve as a
useful \textquotedblleft laboratory\textquotedblright\ or \textquotedblleft
model system\textquotedblright:\ something whose statistical properties are
almost certainly similar to those of \textsc{BosonSampling}, but that's easier
to understand rigorously.

Before going further, we should define \textsc{FermionSampling}\ formally.
\ Let $A\in\mathbb{C}^{m\times n}$\ ($m\geq n$)\ be a column-orthonormal
matrix. \ Then the \textsc{FermionSampling}\ distribution $\mathcal{F}_{A}%
$\ ranges over collision-free outcomes $S\in\Lambda_{m,n}$\ (or equivalently,
subsets of $\left[  m\right]  $\ of size $n$), and is given by%
\[
\Pr_{\mathcal{F}_{A}}\left[  S\right]  =\left\vert \operatorname*{Det}\left(
A_{S}\right)  \right\vert ^{2},
\]
where $A_{S}$\ is the $n\times n$\ submatrix of $A$ corresponding to $S$.
\ Note that, unlike with \textsc{BosonSampling}, we do not need to worry about
outcomes with collisions. \ For if $S\notin\Lambda_{m,n}$, then $A_{S}$\ has a
repeated row, and therefore $\operatorname*{Det}\left(  A_{S}\right)  =0$
(physically, this just reflects the famous \textit{Pauli exclusion principle},
that no two fermions can occupy the same state at the same time).

We now show another difference between \textsc{BosonSampling}\ and
\textsc{FermionSampling}: namely, that the latter is easy to solve
classically. \ As we said, this fact was known
\cite{valiant:qc,td:fermion,knill:matchgate}, but our proof is shorter than
previous ones and leads to a faster algorithm.

\begin{theorem}
\label{fasample}There exists a probabilistic classical algorithm that, given
$A\in\mathbb{C}^{m\times n}$\ as input, samples from $\mathcal{F}_{A}$ in
$O\left(  mn^{2}\right)  $\ time.
\end{theorem}

\begin{proof}
The algorithm is the following:

\begin{enumerate}
\item[(1)] Let $v_{1,i}:=\left(  a_{i1},\ldots,a_{in}\right)  ^{T}$\ be the
column vector in $\mathbb{C}^{n}$ obtained by transposing the $i^{th}$\ row of
$A$.

\item[(2)] For $t=1$\ to $n$:

\begin{itemize}
\item For all $i\in\left[  m\right]  $, let%
\[
p_{t,i}:=\frac{\left\Vert v_{t,i}\right\Vert _{2}^{2}}{n-t+1}.
\]

\item Sample $h_{t}\in\left[  m\right]  $\ from the probability distribution
$\left(  p_{t,1},\ldots,p_{t,m}\right)  $.

\item For all $i\in\left[  m\right]  $, set $v_{t+1,i}$\ equal to the
projection of $v_{t,i}$\ onto the orthogonal complement of $v_{t,h_{t}}$:%
\[
v_{t+1,i}:=v_{t,i}-\frac{v_{t,h_{t}}^{\dagger}v_{t,i}}{\left\Vert v_{t,h_{t}%
}\right\Vert _{2}^{2}}v_{t,h_{t}}.
\]
(So in particular, $v_{t+1,h_{t}}$\ is set to the all-$0$ vector.)
\end{itemize}

\item[(3)] Output $S=\left\{  h_{1},\ldots,h_{n}\right\}  $\ as a sample from
$\mathcal{F}_{A}$.
\end{enumerate}

Clearly the algorithm runs in $O\left(  mn^{2}\right)  $\ time. \ To see that
the probability distributions are normalized, observe that%
\[
\sum_{i=1}^{m}\left\Vert v_{1,i}\right\Vert _{2}^{2}=\sum_{i=1}^{m}\sum
_{j=1}^{n}\left\vert a_{ij}\right\vert ^{2}=n,
\]
and that each iteration decreases $\sum_{i=1}^{m}\left\Vert v_{t,i}\right\Vert
_{2}^{2}$ by exactly $1$. \ Furthermore, if $h_{t}$\ is sampled,\ then%
\[
\left\Vert v_{t,h_{t}}\right\Vert _{2}^{2}=\left(  n-t+1\right)  p_{t,h_{t}%
}>0,
\]
so the projection $v_{t+1,i}$\ is well-defined.

To see that the algorithm is correct, recall that, if $X\in\mathbb{C}^{n}$\ is
an $n\times n$ matrix, then one way to \textit{compute} $\left\vert
\operatorname*{Det}\left(  X\right)  \right\vert ^{2}$\ is to project each row
$x_{t}$ of $X$\ onto the orthogonal complement of the subspace spanned by all
the rows $x_{1},\ldots,x_{t-1}$\ above $x_{t}$, and then take the product of
the squared row-norms of the matrix $X^{\prime}$\ that results. \ But this is
precisely what our sampling procedure does. \ More concretely, let $H=\left(
h_{1},\ldots,h_{n}\right)  \in\left[  m\right]  ^{n}$\ be an ordered sequence
of rows, and let $S=\left\{  h_{1},\ldots,h_{n}\right\}  $.\ \ Then $H$ gets
sampled with probability equal to%
\[
p_{1,h_{n}}\cdots p_{n,h_{n}}=\frac{\left\Vert v_{1,h_{1}}\right\Vert _{2}%
^{2}}{n}\frac{\left\Vert v_{2,h_{2}}\right\Vert _{2}^{2}}{n-1}\cdots
\frac{\left\Vert v_{n,h_{n}}\right\Vert _{2}^{2}}{1}=\frac{\left\vert
\operatorname*{Det}\left(  A_{S}\right)  \right\vert ^{2}}{n!}.
\]
So the probability that $h_{1},\ldots,h_{n}$ get sampled in \textit{any} order
is simply $n!$\ times the above, or $\left\vert \operatorname*{Det}\left(
A_{S}\right)  \right\vert ^{2}$.
\end{proof}

We now turn to understanding the pdf of $\left\vert \operatorname*{Det}\left(
X\right)  \right\vert ^{2}$, where $X\sim\mathcal{N}\left(  0,1\right)
_{\mathbb{C}}^{n\times n}$\ is an iid Gaussian matrix. \ The reason we are
interested in this pdf is that, by Theorem \ref{trunc}, an $n\times n$
submatrix of a Haar-random $A\in\mathbb{C}^{m\times n}$\ will be close in
variation distance to an iid Gaussian matrix if $m\gg n^{2}$. \ And therefore,
the pdf of $\left\vert \operatorname*{Det}\left(  X\right)  \right\vert ^{2}%
$\ controls the behavior of a Haar-random \textsc{FermionSampling}%
\ distribution, in exactly the same way that we saw that the pdf of
$\left\vert \operatorname*{Per}\left(  X\right)  \right\vert ^{2}$\ controls
the behavior of a Haar-random \textsc{BosonSampling}\ distribution.
\ (However, we will not discuss explicitly how to move from statements about
$\left\vert \operatorname*{Det}\left(  X\right)  \right\vert ^{2}$\ to
statements about \textsc{FermionSampling}, since it is precisely analogous to
how we moved from statements about $\left\vert \operatorname*{Per}\left(
X\right)  \right\vert ^{2}$\ to statements about \textsc{BosonSampling}\ in
Section \ref{DEV}.)

The key to understanding the pdf of $\left\vert \operatorname*{Det}\left(
X\right)  \right\vert ^{2}$\ is the following proposition, which is noted for
example in Costello and Vu \cite[Appendix]{cv}, and which we also used in
\cite{aark}.

\begin{proposition}
\label{piped}If $\mathcal{N}\left(  0,1\right)  _{\mathbb{C}}^{n\times n}$,
then $\left\vert \operatorname*{Det}\left(  X\right)  \right\vert ^{2}$\ is
distributed precisely as%
\[
\left\vert y_{11}\right\vert ^{2}\left(  \left\vert y_{21}\right\vert
^{2}+\left\vert y_{22}\right\vert ^{2}\right)  \cdots\left(  \left\vert
y_{n1}\right\vert ^{2}+\cdots+\left\vert y_{nn}\right\vert ^{2}\right)  ,
\]
where each $y_{ij}$\ is an independent $\mathcal{N}\left(  0,1\right)
_{\mathbb{C}}$\ Gaussian.
\end{proposition}

\begin{proof}
[Proof Sketch]Let $v_{1},\ldots,v_{n}$\ be the rows of $X$, and let $w_{t}%
$\ be the projection of $v_{t}$\ onto the orthogonal complement of the
subspace spanned by $v_{1},\ldots,v_{t-1}$. \ Then just like in the proof of
Theorem \ref{fasample}, we can write $\left\vert \operatorname*{Det}\left(
X\right)  \right\vert ^{2}$\ as $\left\Vert w_{1}\right\Vert _{2}^{2}%
\cdots\left\Vert w_{n}\right\Vert _{2}^{2}$. \ Now, clearly $\left\Vert
w_{1}\right\Vert _{2}^{2}=\left\Vert v_{1}\right\Vert _{2}^{2}$\ is
distributed as a complex $\chi^{2}$\ random variable with $n$ degrees of
freedom. \ But it follows, from the rotational symmetry of the Gaussian
measure, that $\left\Vert w_{2}\right\Vert _{2}^{2}$\ is distributed as a
complex $\chi^{2}$\ with $n-1$\ degrees of freedom, $\left\Vert w_{3}%
\right\Vert _{2}^{2}$\ is distributed as a complex $\chi^{2}$\ with
$n-2$\ degrees of freedom, and so on. \ Moreover these $\chi^{2}$'s are all
independent of each other, since the $v_{t}$'s are.
\end{proof}

Using Proposition \ref{piped}, Girko \cite{girko} and Costello and Vu
\cite[Appendix]{cv}\ showed that, if $X\sim\mathcal{N}\left(  0,1\right)
_{\mathbb{R}}^{n\times n}$, then%
\[
\frac{\ln\left\vert \operatorname*{Det}\left(  X\right)  \right\vert -\ln
\sqrt{\left(  n-1\right)  !}}{\sqrt{\frac{\ln n}{2}}}%
\]
converges weakly to the standard normal distribution $\mathcal{N}\left(
0,1\right)  _{\mathbb{R}}$. \ This result falls short of what we want in two
minor respects: it's for real rather than complex $X$, and it doesn't
quantitatively bound the convergence rate. \ For completeness, we now fill
these gaps. \ We will do so using the Berry-Esseen Theorem (Theorem
\ref{berryesseen}), but in a variant for sums of non-iid random variables.

\begin{theorem}
[Berry-Esseen, non-iid case]\label{noniid}Let $Z_{1},\ldots,Z_{n}$\ be real
iid random variables satisfying%
\begin{align*}
\operatorname{E}\left[  Z_{i}\right]   &  =\upsilon_{i},\\
\operatorname{E}\left[  \left(  Z_{i}-\upsilon_{i}\right)  ^{2}\right]   &
=\sigma_{i}^{2}>0,\\
\operatorname{E}\left[  \left\vert Z_{i}-\upsilon_{i}\right\vert ^{3}\right]
&  =\rho_{i}<\infty.
\end{align*}
Then let%
\[
Z:=Z_{1}+\cdots+Z_{n},
\]
and let $W$\ be a real Gaussian with mean $\upsilon_{1}+\cdots+\upsilon_{n}%
$\ and variance $\sigma_{1}^{2}+\cdots+\sigma_{n}^{2}$. \ Then for all
$x\in\mathbb{R}$,%
\[
\left\vert \Pr\left[  Z\leq x\right]  -\Pr\left[  W\leq x\right]  \right\vert
\leq C\frac{\rho_{1}+\cdots+\rho_{n}}{\left(  \sigma_{1}^{2}+\cdots+\sigma
_{n}^{2}\right)  ^{3/2}},
\]
where $C$ is some universal constant.
\end{theorem}

By combining Theorem \ref{noniid} with Lemma \ref{heartbreak}\ and Proposition
\ref{piped}, we can now upper-bound the variation distance between $\left\vert
\operatorname*{Det}\left(  X\right)  \right\vert ^{2}$\ and a lognormal random variable.

\begin{theorem}
\label{detdist}Let $X\sim\mathcal{N}\left(  0,1\right)  _{\mathbb{C}}^{n\times
n}$\ be Gaussian. \ Then for all $x\in\mathbb{R}$,%
\[
\Pr_{X}\left[  \frac{\left\vert \operatorname*{Det}\left(  X\right)
\right\vert ^{2}}{n!}\leq\sqrt{\frac{e}{2\pi n}}e^{x\sqrt{\ln n+1+\gamma}%
}\right]  =\int_{-\infty}^{x}\frac{e^{-y^{2}/2}}{\sqrt{2\pi}}dy\pm O\left(
\frac{1}{\log^{3/2}n}\right)  .
\]
In other words, the cdf of $\left\vert \operatorname*{Det}\left(  X\right)
\right\vert ^{2}$\ is pointwise $O\left(  \log^{-3/2}n\right)  $-close to the
cdf of a lognormal random variable.
\end{theorem}

\begin{proof}
We have%
\begin{align*}
\ln\left(  \left\vert \operatorname*{Det}\left(  X\right)  \right\vert
^{2}\right)   &  =\ln\left(  \left\vert y_{11}\right\vert ^{2}\right)
+\ln\left(  \left\vert y_{21}\right\vert ^{2}+\left\vert y_{22}\right\vert
^{2}\right)  +\cdots+\ln\left(  \left\vert y_{n1}\right\vert ^{2}%
+\cdots+\left\vert y_{nn}\right\vert ^{2}\right) \\
&  =\ell_{1}+\cdots+\ell_{n},
\end{align*}
where the first line uses the notation of Proposition \ref{piped} (each
$y_{ij}$\ being an independent $\mathcal{N}\left(  0,1\right)  _{\mathbb{C}}%
$\ Gaussian), and the second line uses the notation of Section \ref{DETECT}.
\ Thus, $\ln\left(  \left\vert \operatorname*{Det}\left(  X\right)
\right\vert ^{2}\right)  $\ is similar to the random variable $L=\ln
R$\ studied in Section \ref{DETECT}, except that, whereas $L$ was a sum of $n$
\textit{identical} $\ell_{n}$\ random variables, $\ln\left(  \left\vert
\operatorname*{Det}\left(  X\right)  \right\vert ^{2}\right)  $\ is a sum of
$n$ non-identical variables $\ell_{1},\ldots,\ell_{n}$.

Still, if we let%
\begin{align*}
\upsilon_{t}  &  =\operatorname{E}\left[  \ell_{t}\right]  ,\\
\sigma_{t}^{2}  &  =\operatorname*{Var}\left[  \ell_{t}\right]  ,\\
\rho_{t}  &  =\operatorname{E}\left[  \left\vert \ell_{n}-\upsilon
_{t}\right\vert ^{3}\right]  ,
\end{align*}
then by Lemma \ref{heartbreak},%
\[
\sigma_{t}^{2}=\frac{1+o\left(  1\right)  }{t},~~~~\rho_{t}=\frac
{3^{3/4}+o\left(  1\right)  }{t^{3/2}}.
\]
So let $W$\ be a real Gaussian with mean $\upsilon_{1}+\cdots+\upsilon_{n}%
$\ and variance $\sigma_{1}^{2}+\cdots+\sigma_{n}^{2}$. \ Then by Theorem
\ref{noniid}, for all $x\in\mathbb{R}$,%
\[
\left\vert \Pr\left[  \ln\left(  \left\vert \operatorname*{Det}\left(
X\right)  \right\vert ^{2}\right)  \leq x\right]  -\Pr\left[  W\leq x\right]
\right\vert \leq C\frac{\rho_{1}+\cdots+\rho_{n}}{\left(  \sigma_{1}%
^{2}+\cdots+\sigma_{n}^{2}\right)  ^{3/2}}=O\left(  \frac{1}{\log^{3/2}%
n}\right)  .
\]
It remains only to estimate the mean and variance of $\ln\left(  \left\vert
\operatorname*{Det}\left(  X\right)  \right\vert ^{2}\right)  $. \ Using Lemma
\ref{heartbreak}, it is not hard to show the following:%
\begin{align*}
\operatorname*{E}\left[  \ln\left(  \left\vert \operatorname*{Det}\left(
X\right)  \right\vert ^{2}\right)  \right]   &  =\upsilon_{1}+\cdots
+\upsilon_{n}=n\ln n-n+\frac{1}{2}-O\left(  \frac{1}{n}\right)  ,\\
\operatorname*{Var}\left[  \ln\left(  \left\vert \operatorname*{Det}\left(
X\right)  \right\vert ^{2}\right)  \right]   &  =\sigma_{1}^{2}+\cdots
+\sigma_{n}^{2}=\ln n+1+\gamma+O\left(  \frac{1}{n^{2}}\right)  .
\end{align*}
Since the $O\left(  1/n\right)  $\ and $O\left(  1/n^{2}\right)  $\ error
terms are \textquotedblleft swamped\textquotedblright\ by the $O\left(
\log^{-3/2}n\right)  $, this gives us that for all $x\in\mathbb{R}$,%
\[
\Pr_{X}\left[  \frac{\ln\left(  \left\vert \operatorname*{Det}\left(
X\right)  \right\vert ^{2}\right)  -\left(  n\ln n-n+1/2\right)  }{\sqrt{\ln
n+1+\gamma}}\leq x\right]  =\int_{-\infty}^{x}\frac{e^{-y^{2}/2}}{\sqrt{2\pi}%
}dy\pm O\left(  \frac{1}{\log^{3/2}n}\right)  .
\]
Rearranging and applying Stirling's approximation now yields the theorem.
\end{proof}

As a counterpoint to Theorem \ref{detdist}, let us give a simple argument for
why $\left\vert \operatorname*{Det}\left(  X\right)  \right\vert ^{2}$\ cannot
be \textit{exactly} lognormal, for any fixed $n$. \ Recall, from Section
\ref{DEV}, that the pdf of $\left\vert \operatorname*{Per}\left(  X\right)
\right\vert ^{2}$\ is a mixture of exponential random variables (Lemma
\ref{expmix}) and is therefore monotonically decreasing (Theorem
\ref{grabbag}). \ It is easy to see that the same is true for $\left\vert
\operatorname*{Det}\left(  X\right)  \right\vert ^{2}$, since the proof of
Lemma \ref{expmix}\ works just as well for the determinant as for the
permanent. \ By contrast, the pdf of a lognormal random variable is
\textit{not} monotonically decreasing,\ but is $0$ at the origin.

\end{document}